\documentclass[aps,showpacs,pra,twocolumn,superscriptaddress]{revtex4-1}
\usepackage[utf8]{inputenc}
\usepackage{braket}
\usepackage{bbm}
\usepackage{amsmath}
\usepackage{amssymb}
\usepackage{amsthm}
\newtheorem{thm}{Theorem}

\usepackage{graphicx} 
\usepackage[caption=false]{subfig}
\usepackage{float}
\usepackage[a4paper, total={6.5in, 10in}]{geometry}
\usepackage{xcolor}
\usepackage{mathtools}
\usepackage[colorlinks=true,citecolor=blue,linkcolor=blue,urlcolor=blue]{hyperref}
\usepackage{tabularx}

\theoremstyle{plain}

\theoremstyle{definition}

\newcommand{\aqa}{$\langle aQa ^L\rangle $ Applied Quantum Algorithms, Universiteit Leiden}
\newcommand{\lorentz}{Instituut-Lorentz, Universiteit Leiden, Niels Bohrweg 2, 2333 CA Leiden, Netherlands}
\newcommand{\liacs}{LIACS, Universiteit Leiden, Niels Bohrweg 1, 2333 CA Leiden, Netherlands}

\begin{document}

\setlength{\abovedisplayskip}{6pt}
\setlength{\belowdisplayskip}{6pt}
\setlength{\abovedisplayshortskip}{6pt}
\setlength{\belowdisplayshortskip}{6pt}

\title{Stoquastic permutationally invariant Bell operators}

\author{Jan Li}
\affiliation{\aqa}
\affiliation{\lorentz}

\author{Owidiusz Makuta}
\affiliation{\aqa}
\affiliation{\lorentz}

\author{Evert van Nieuwenburg}
\affiliation{\aqa}
\affiliation{\lorentz}
\affiliation{\liacs}

\author{Jordi Tura}
\affiliation{\aqa}
\affiliation{\lorentz}

\date{\today}

\begin{abstract}
As Hermitian operators, many-body Bell operators can naturally be identified as many-body Hamiltonians.
An important subclass of such Hamiltonians is the stoquastic class, characterized by having nonpositive off-diagonal matrix elements in a given basis.
Interestingly, this property is shared by the permutationally invariant (PI) Bell operators underlying the largest Bell-correlation experiments to date.
In this work, we explore the connection between many-body PI Bell operators and stoquasticity. 
We introduce the stoquasticity cone, which allows for a full characterization of the stoquastic parameter regimes for any PI Bell operator. 
We use this to show that PI Bell operators of the binary-input binary-output scenario consisting of at most three-body correlators can always be rendered stoquastic for any set of measurement parameters. 
Additionally, we also provide examples that use the stoquasticity cone to optimize for the quantum-classical gap.
Numerical evidence suggests that the Bell operator used in the largest experiments to date is optimal with respect to stoquasticity.  
To the best of our knowledge, this work establishes the first connection between PI Bell operators and stoquasticity.
\end{abstract}

\maketitle

\section{Introduction}

Bell nonlocality has been of foundational and philosophical interest in physics~\cite{bell_einstein_1964} and it also serves as the key resource for tasks such as quantum key distribution~\cite{primaatmaja_security_2023}, randomness amplification and expansion~\cite{colbeck_free_2012, pironio_random_2010,vazirani_certifiable_2012}, and self-testing~\cite{supic_self-testing_2020}. Moreover, such tasks can be framed within the device-independent paradigm~\cite{acin_certified_2016,gallego_device-independent_2010}, which allows for certification of quantum behavior from a minimal set of assumptions.
Central to the device-independent paradigm are nonlocal correlations. 
Experimental verification of such correlations relies on measuring the expected value of a Bell operator~\cite{brunner_bell_2014}  and checking whether it exceeds a classical threshold.

Extending Bell certification to the many-body regime is both timely and challenging. On one hand, recent experimental progress in controlling large quantum systems~\cite{bluvstein_quantum_2022,mi_stable_2024} makes such certifications increasingly relevant. On the other hand, while multipartite Bell operators have been studied for a long time~\cite{mermin_extreme_1990,toth_two-setting_2006,aolita_fully_2012}, their construction and implementation pose a significant challenge due to the increased complexity of the scenario. 
Classes of Bell operators that are experimentally friendly have been developed in the last decade, often relying on low-order correlators and symmetric expressions~\cite{tura_nonlocality_2015,wang_entanglement_2017,tura_detecting_2014,tura_nonlocality_2015,collins_bell_2002,salavrakos_bell_2017,baccari_scalable_2020}.
Indeed, for Bell experiments in the largest system sizes to date, $480$ ${}^{87}\mathrm{Rb}$ atoms in a Bose-Einstein condensate~\cite{schmied_bell_2016} and $5\cdot 10^5$ atoms in a thermal ensemble~\cite{engelsen_bell_2017}, such Bell operators were used.

Demonstrating Bell correlations in systems of this size requires preparing many-body entangled states that are both sufficiently entangled and experimentally accessible with current technology. 
\textit{Stoquastic} Hamiltonians, which are Hamiltonians with only nonpositive off-diagonal elements for a given basis, may be particularly well-suited for this purpose. 
Stoquasticity was introduced in~\cite{bravyi_merlin-arthur_2006} from a complexity-theoretic point of view. 
This property was used in Ref.~\cite{bravyi_complexity_2010} to find the first natural complete problem of the classical complexity class MA, which is a probabilistic generalization of NP. 
Another interesting feature is that stoquastic Hamiltonians do not have the sign problem in Monte Carlo sampling methods~\cite{troyer_computational_2005,bravyi_complexity_2008}, making them more amenable to classical simulation. 
Later, in a similar spirit, it was shown that tensor network contractions become easier on average when there is a positive bias in the entries~\cite{jiang_positive_2024,chen_sign_2025}.
Additionally, it has been shown that Hamiltonians without the sign problem potentially allow for a subexponential speed up in adiabatic quantum computing compared to classical computing~\cite{gilyen_subexponential_2021, hastings_power_2021}. 

Interestingly, it has so far gone unnoticed that stoquasticity is also related to nonlocality. 
More precisely, the Bell operators used in the aforementioned experiments~\cite{schmied_bell_2016,engelsen_bell_2017} can be made stoquastic using solely local unitary transformations. 
This raises the natural question of how general this property is and which Bell correlations remain accessible when restricting to stoquastic operators, which is the central question of this work.

In this work, we aim to characterize the conditions under which Bell operators are stoquastic in the locally binary-input and binary-output scenario.
We focus on a class of permutationally invariant (PI) Bell operators, which are characterized by Bell coefficients $\vec{\alpha}$ and the measurement parameters $\vec{\theta}$, and we consider the case where the set of measurements is the same for all parties. 
In particular, we focus on the symmetric subspace, which is the subspace of 
states invariant under permutation of the $n$ qubits, as the ground state of two-body PI Hamiltonians is expected to lie there~\cite{fadel_bell_2018} and as the analysis in the other subspaces is qualitatively similar~\cite{tura_nonlocality_2015}.

Our first contribution is that for a class of PI Bell inequalities, the corresponding operators can always be made stoquastic under suitably chosen measurement parameters and we provide the corresponding conditions needed for them. 
Our second main contribution is a full characterization of the set of Bell coefficients for which the corresponding Bell operator is stoquastic under fixed measurement parameters $\vec{\theta}$.
We show that this set is both a polyhedron and a cone, and thus we name it the stoquasticity cone. 

Using the stoquasticity cone, we show that Bell operators limited to at most three-body operators can always be made stoquastic. 
Additionally, we numerically optimize the quantum-classical gap over all stoquastic Bell coefficients at fixed measurement parameters. 
We conclude from numerical optimizations that the Bell operator used for the largest many-body Bell violations~\cite{schmied_bell_2016,engelsen_bell_2017} is optimal with respect to stoquasticity at its optimal measurement parameters and proves to be a good candidate for other measurement parameters. 

Finally, we show that with any state with non-negative amplitudes in the Dicke basis, we can associate a stoquastic PI Bell operator whose ground state is that state. As the squared amplitudes of such a state define a probability distribution over the Dicke basis, stoquastic PI Bell operators can serve as parent Hamiltonians for arbitrary such probability distributions.
This potentially comes at the cost of requiring operators of order $K \sim \mathcal{O}(n)$.
However for values of $k \leq 3$, it is already possible to capture superpositions of Dicke states with a Gaussian profile with the variance scaling at $\sqrt{n}$, which correspond to spin squeezed states that are routinely prepared in experiments~\cite{marconi_symmetric_2026}.
This highlights the importance of higher-order operators in accessing the full range of ground state probability distributions.

The remainder of this paper is organized as follows. 
In Sec.~\ref{Sec:prelim}, we introduce the necessary preliminaries, including permutationally invariant Bell operators, stoquastic Hamiltonians and the relevant notions from polyhedral theory. 
In Sec.~\ref{sec:MotEx}, we motivate our study by considering the example of the Bell operator used in the largest experiments to date. 
In Sec.~\ref{sec:GenTwoBodyClass}, we analyze the stoquasticity of a class of two-body PI Bell operators and establish conditions under which they can be made stoquastic.
In Sec.~\ref{sec:StoqCone}, we introduce the stoquasticity cone and in Sec.~\ref{sec:2BodyRes} we provide its full characterization for the two-body case.
In Sec.~\ref{sec:3BodyRes}, we extend the analysis to three-body PI Bell operators and show that they can always be made stoquastic.
In Sec.~\ref{sec:HighOrderDistr}, we discuss the connection between stoquastic PI Bell operators and arbitrary probability distributions over the Dicke basis.
Finally, we conclude and provide an outlook in Sec.~\ref{sec:Conclusion}.

\section{preliminaries}\label{Sec:prelim}

As the study of many-body Bell inequalities is hindered by the exponential scaling of the probability parameter space with the number of qubits $n$, simplifying assumptions are often introduced to reduce the complexity.
In this work, we focus on permutationally invariant operators.
Additionally, we restrict ourselves to a scenario in which each party can choose one of two measurement inputs $x \in \{0,1\}$, and gets as an outcome one of the two possible outputs $a \in \{-1,+1\}$. With each input $x$ we associate a dichotomic measurement operator $\mathbf{M}_{x}$ with eigenvalues $\pm 1$. Importantly, due to Jordan's Lemma~\cite[Lemma~2]{toner_monogamy_2006}, we can restrict $\mathbf{M}_{x}$ to be a $2\times2$ real matrix, which allows us to parametrize them as 
\begin{equation}
\begin{aligned}
\mathbf{M}_0=\cos(\varphi) Z+ \sin(\varphi) X,\\ \mathbf{M}_1=\cos(\theta) Z+ \sin(\theta) X,
\label{eq:QubitMeas}
\end{aligned}
\end{equation}
where $Z$ and $X$ denote the Pauli matrices and $\varphi, \theta$ denoting the measurement parameters of the corresponding measurement settings.

Using these local measurements $M_{x}$ we define the $K$-body PI measurement operator as
\begin{equation}
     \mathbf{S}_{\vec{x}} := \sum_{\substack{i_1,\ldots,i_K=0\\ \text{all distinct}}}^{n-1} 
     \bigotimes_{j=1}^{K} \mathbf{M}_{x_j}^{i_j},
     \label{eq:kPIMeasOP}
\end{equation}
where $\vec{x} = (x_1,\ldots,x_K) \in \{0,1\}^K$ is the vector containing the measurement settings of all parties, $i_j \in \{0,\ldots,n-1\}$ denotes the qubit on which $\mathbf{M}_{x_j}^{i_j}$ acts, and the sum runs over all $K$-tuples of 
distinct qubits. 
We call $K$ the order of the PI measurement operator.
The explicit forms for first order and second order PI measurement operators are 
\begin{equation}
    \mathbf{S}_{x_{1}} := \sum_{i_1=0}^{n-1} \mathbf{M}_{x_1}^{i_1}, \nonumber
\quad
\mathbf{S}_{x_{1},x_{2}}:=\sum_{\substack{i_1,i_2=0\\i_1\neq i_2}}^{n-1} \mathbf{M}_{x_1}^{i_1} \otimes \mathbf{M}_{x_2}^{i_2},\nonumber
\end{equation}
respectively.

A PI Bell operator is defined as a linear combination of PI measurement operators
\begin{equation}
    \mathbf{B} \coloneqq \sum_{\vec{x} \in \mathcal{X}_K}\alpha_{\vec{x}}  \mathbf{S}_{\vec{x}},
    \nonumber
\end{equation}
where $\mathcal{X}_K := \bigcup_{k=1}^{K} \{0,1\}^k$,  $\alpha_{\vec{x}}  \in \mathbbm{R}$ is the corresponding Bell coefficient of the PI measurement operator $\mathbf{S}_{\vec{x}}$. 
We say that a PI Bell operator is of order $K$ if the largest PI measurement operator is of order $K$.
In the remainder of this work, we will often label the index of the Bell coefficients as $i \in (0,1,2,3,4,5,6,7,...)$  instead of the corresponding $\vec{x} \in (0,1,00,01,11,000,001,011,...)$, writing $\mathbf{B} = \sum_{i} \alpha_{i}  \mathbf{S}_{i}$ where $\mathbf{S}_{i} \equiv \mathbf{S}_{\vec{x}}$.
For example, the two-body PI Bell operator is given by
\begin{equation}
 \label{eq:TwoBodyBellOp}
 \mathbf{B} = \alpha_0  \mathbf{S}_0 + \alpha_1  \mathbf{S}_1 + \alpha_2 \mathbf{S}_{00} + \alpha_3  \mathbf{S}_{01} + \alpha_4 \mathbf{S}_{11}. 
\end{equation}

\subsection{Block decomposition of permutationally invariant Bell operators}\label{sec:PIBlock}

Due to the permutational symmetry, the Bell operators can be decomposed in a block-diagonal form, with each subblock $\mathbf{B}_J$ labelled by $J \in \{J_0, J_0 + 1, \ldots, n/2\}$, where $J_0 =0$ if $n$ is even and $J_0 =1/2$ if $n$ is odd~\cite{moroder_permutationally_2012}.
See Fig. \ref{fig:rhoPI_block} for a schematic representation. 
We call the $\mathbf{B}_J$ the $J$-th block of $\mathbf{B}$. 
The basis of the $J$-th block is given by
\begin{equation}
\label{eq:DickeBasis}
\{\ket{D^{k}_{2J}} \otimes \ket{\psi^-}^{\otimes m/2}\}_{k=0}^{2J},
\end{equation}
where $m = n-2J$, $\ket{\psi^-} \coloneqq \frac{\ket{01}-\ket{10}}{\sqrt{2}}$ and 
\begin{equation}
\label{eq:defDicke}
\ket{D^k_{2J}}:={2J\choose k}^{-1/2}\sum_{\sigma\in\Pi_{2J}}V_{\sigma}\left(\ket{0}^{\otimes 2J-k}\ket{1}^{\otimes k}\right),
\nonumber
\end{equation}
and $\Pi_{2J}$ is the symmetric group on $2J$ elements and $V_{\sigma}$ is defined as
\begin{equation}
V_\sigma \left( \bigotimes_{i=0}^{2J-1} \ket{\psi_i} \right) = \bigotimes_{i=0}^{2J-1} \ket{\psi_{\sigma^{-1}(i)}}, \nonumber
\end{equation}
where $\sigma \in \Pi_{2J}$. 
The state $\ket{D^k_{2J}}$ is called the Dicke state of $2J$ qubits with $k$ excitations~\cite{dicke_coherence_1954,dur_three_2000}.
In Ref.~\cite{tura_nonlocality_2015}, it was shown that the $J$-th block $\mathbf{B}_J$ of the Bell operator \eqref{eq:TwoBodyBellOp} is of the form 
\begin{equation}
\begin{aligned}
\label{eq:penta-block}
  (\mathbf{B}_J)_{i,j} = d_i \delta_{i,j} &+ u_i \delta_{i,j-1}+u_j
\delta_{i-1,j} \\
&+ v_i \delta_{i,j-2} + v_j\delta_{i-2,j},
\end{aligned}
\end{equation}
where $(\mathbf{B}_J)_{i,j}$ represents the matrix element of $\mathbf{B}_J$ with $i,j$ running between $0$ and $2J$, the coefficients $d_i, u_i, v_i \in \mathbbm{R}$ and $\delta_{i,j}$ denotes the Kronecker delta. 
The fact that only the diagonal elements, off-diagonal elements and second off-diagonal elements are nonzero makes it particularly interesting for our stoquasticity analysis, as all other elements are zero and automatically satisfy the stoquasticity conditions. 

More generally, it was shown in Ref.~\cite{tura_nonlocality_2015} that $K$-th order PI measurement operators can only contribute to at most the $K$-th off-diagonal.
This is due to the fact that $S_{\vec{x}}$ of order $K$ can be decomposed into a sum of products of at most $K$ single-body PI Pauli operators defined as $S_X \coloneq \sum_{i=0}^{n-1} X^i$, $S_Y \coloneq \sum_{i=0}^{n-1} Y^i$ and $S_Z \coloneq \sum_{i=0}^{n-1} Z^i$.
The operator $S_X \ket{D_k}$ generates a superposition of $\ket{D_{k-1}}$ and $\ket{D_{k+1}}$ and similarly for $S_Y \ket{D_k}$.
The operator $S_Z$ has only diagonal terms in the Dicke basis.
Cf. App. \ref{app:HighToSingle} for the exact expressions of $S_X$, $S_Y$ and $S_Z$ in the Dicke basis and the decomposition for two-body and three-body PI measurement operators into them.
Therefore, products of $K$ single-body PI Pauli operators and consequently $K$-th order PI measurement operators can reach at most the $K$-th off diagonal. 

For our analysis, we mainly focus on the symmetric subspace, i.e. the block 
$\mathbf{B}_J$ with $J=n/2$.
This is physically motivated by the fact that for two-body PI Hamiltonians, 
a Holstein-Primakoff argument suggests that the ground state is expected to 
lie in the symmetric block~\cite{fadel_bell_2018}.
Additionally, focusing on smaller blocks $\mathbf{B}_J$ with $J < n/2$ does 
not qualitatively change the analysis.
The matrix elements of $\mathbf{B}_J$ 
take the same algebraic form as those of $\mathbf{B}_{n/2}$, with $n$ replaced 
by $2J$~\cite{tura_nonlocality_2015}, making smaller blocks equivalent to 
the symmetric block at a reduced system size.

\subsection{Stoquastic operators}

A Hermitian operator $\mathbf{H}$ is said to be stoquastic with respect to a basis 
$\mathcal{B} = \{\ket{i}\}$ if 
\begin{equation}
    \bra{i} \mathbf{H} \ket{j} \leq 0 \quad \forall \; i \neq j.
\end{equation}
Stoquasticity was first introduced in~\cite{bravyi_merlin-arthur_2006} 
in the context of $k$-local Hamiltonians, where $\mathbf{H} = \sum_x \mathbf{H}_x$ and each 
local term $\mathbf{H}_x$ acts on at most $k$ qubits.

An important consequence of stoquasticity is that the ground state $\ket{\psi_0}$ of a 
stoquastic Hamiltonian has non-negative amplitudes in the basis $\mathcal{B}$, i.e. 
$\braket{i|\psi_0} \geq 0$ for all $i$.
This follows from the Perron-Frobenius theorem applied to the matrix $c\mathbbm{I} - \mathbf{H}$ for a sufficiently large $c$, compensating for the potentially negative diagonal elements of $-\mathbf{H}$. 
The ground state of $\mathbf{H}$ coincides with the eigenvector corresponding to the largest eigenvalue of $c\mathbbm{I} - \mathbf{H}$, which by Perron-Frobenius has nonnegative entries.
As a consequence, the ground state can be interpreted as a probability distribution over 
the basis states $\mathcal{B}$.

In the context of Bell operators, ground 
states play a particularly important role as they correspond to the states achieving the maximal quantum violations. 
As Bell operators are Hermitian, they fall within the general framework of stoquastic 
operators, and the question of whether a Bell operator can be made stoquastic via a local 
unitary transformation is therefore well-posed.
Though any Hermitian operator is trivially stoquastic in its eigen basis, the complexity of finding a local stoquastic basis grows exponentially with the system size, 
making the characterization of stoquastic Bell operators a nontrivial problem~\cite{marvian_computational_2019, ioannou_termwise_2022}.

\subsection{Polyhedra and cones}\label{sec:Polyhedron}

In the Bell scenario, the set of local correlations forms a polytope~\cite{brunner_bell_2014}. Interestingly, a related structure called a polyhedron can be used for the characterization of stochasticity.

A polytope is by construction a bounded object. 
If the object has at least one direction in which it is unbounded, it is called a polyhedron. 
A polyhedron can be characterized by the set of bounding hyperplanes,
\begin{align}
P &= \bigcap_{j=1}^k \left\{ x \in \mathbb{R}^d \ \middle|\ a_j^\top x \le b_j \right\},
\nonumber
\end{align}
Let $H_j$ denote the half space associated with $a_j^\top x \le b_j$.
Not all hyperplanes are equally informative. 
The corresponding hyperplane 
$a_j^\top x = b_j$ is redundant if all points satisfying the remaining half-spaces 
$H_i$ automatically satisfy $H_j$ as well, meaning it does not contribute any 
additional constraint on the feasible region.
A hyperplane is called irredundant if it is not redundant, i.e. it actively 
constrains the feasible region.

Alternatively, a polyhedron can also be characterized by a set of vertices and vectors which represent the direction in which it is unbounded
\begin{equation}
\begin{aligned}
P = \Bigl\{
\sum_{i=1}^I \lambda_i v_i + \sum_{j=1}^J \mu_j r_j + 
\sum_{k=1}^K  \nu_k l_k \;\Bigm|\;
\\
\lambda_i \ge 0,\
\mu_j \ge 0,\
\nu_k \in \mathbbm{R},\
\sum_{i=1}^I \lambda_i = 1
\Bigr\}.
\label{eq:Vpolyhedron}
\end{aligned}
\end{equation}
Here $\{v_i\}$ are the set of vertices of the polyhedron, $\{r_j\}$ are the set of extremal rays and $\{l_k\}$ are the set of extremal lines. 
These vertices, extremal rays and extremal lines are also called the generators of the cone. 
The extremal rays and extremal lines characterize the directions in which the polyhedron is unbounded. 
It has been proven that the hyperplane and vertex descriptions are equivalent, and therefore either representation can always be converted into the other~\cite{ziegler_lectures_1995}. 
We will call the former the primal description and the latter the dual description.

A particularly important special case arises when $b_j = 0$ for all $j$, 
i.e. when the system of inequalities is homogeneous. 
The resulting object is called a cone $C$, and as we will show in Sec.~\ref{sec:StoqCone}, 
the set of Bell coefficients yielding stoquastic operators naturally takes this form.

For any cone $C$ defined by a homogeneous system of inequalities, its dual description simplifies considerably. 
In the primal description, homogeneity implies that if $x \in C$ then 
$\lambda x \in C$ for all $\lambda \geq 0$. 
This implies that any nonzero feasible point $x$ can be written as 
$x = \frac{1}{2} \cdot 0 + \frac{1}{2} \cdot 2x$, i.e. as a convex combination 
of two distinct feasible points, and is therefore not a vertex. 
The dual description of a cone $C$ thus contains no nonzero vertices, and reduces 
to a combination of extremal rays and lines. 
When $C$ is pointed, i.e. $C \cap -C = \{0\}$, there are additionally no lines, 
since a line $l$ would require both $l \in C$ and $-l \in C$, contradicting 
pointedness. 
The origin is the only vertex of a pointed cone, and its dual description 
consists of extremal rays only.

When given the primal description of a pointed cone, the extremal rays of the dual description can be found systematically. 
Let us characterize the primal description of a polyhedron by the matrix $A \in \mathbbm{R}^{m \times d}$ where the rows $a_j^T \in \mathbbm{R}^d$ correspond to the hyperplanes $a_j^T x = 0$. 
If the cone is pointed, then the rays can be found through intersecting $d-1$ linearly independent hyperplanes and checking if all points of said intersection satisfy all $m$ hyperplane conditions. 
This is because the intersection of $d-1$ linearly independent hyperplanes gives rise to a one-dimensional object, which may be represented by a vector. 
Any positive scalar multiplication of the vector will still lie inside the cone if the original vector does. 
Conversely, all extremal rays have to lie on the intersection of the $d-1$ independent hyperplanes.
If this were not the case, then the ray could be written as a combination of points on the hyperplanes, making it therefore not extremal. 

\section{A motivating Example}\label{sec:MotEx}

For the detection of the largest many-body Bell correlations to date~\cite{schmied_bell_2016,engelsen_bell_2017}, the following Bell inequality was used
\begin{equation}
 \label{eq:Ineq6}
 \begin{aligned}
 \bigg[- 2 {\cal S}_0(\varphi,\theta) &+ \frac{1}{2}{\cal S}_{00}(\varphi,\theta) - {\cal S}_{01}(\varphi,\theta) \\ 
 &+ \frac{1}{2}{\cal S}_{11}(\varphi,\theta) + 2N \bigg] \geq 0,
  \end{aligned}
\end{equation}
which was introduced in Ref.~\cite{tura_nonlocality_2015}. 
Here, $\mathcal{S}_{\vec{x}} \coloneqq  \langle \mathbf{S}_{\vec{x}}\rangle$ denotes the expected value of the PI measurement operator, which is built from the local measurement operators of Eq. \eqref{eq:QubitMeas} with parameters $\varphi,\theta$. 
The maximum violation of this inequality is achieved at $(\varphi,\theta)=(\pi/6,5\pi/6)$, see Eq. \eqref{eq:GaussianDickeCoeff} for the corresponding state.
Interestingly, at these measurement parameters the corresponding Bell operator
\begin{equation}
 \label{eq:Ineq6Op}
 \begin{aligned}
 \mathbf{B} = - 2 \mathbf{S}_0 + \frac{1}{2}\mathbf{S}_{00} - \mathbf{S}_{01} + \frac{1}{2}\mathbf{S}_{11},
 \end{aligned}
\end{equation}
is stoquastic in the Dicke basis, i.e. all off-diagonal elements of the operator are nonpositive~\cite{tura_nonlocality_2015}.

As stoquasticity is a basis dependent property, we remark that stoquasticity is commonly defined with respect to the computational basis. 
The following theorem shows that stoquasticity in the computational basis and stoquasticity in the Dicke Basis are essentially equivalent however. 
Stoquasticity of $\mathbf{B}_J$ in the Dicke basis is sufficient to guarantee 
stoquasticity in the computational basis, up to the addition of a projector 
onto the symmetric subspace.
\begin{thm}\label{thm:DickeToComp}
    Let $\mathbf{B}_J$ be a Bell operator supported on the symmetric subspace that 
    is stoquastic in the Dicke basis, i.e. $(\mathbf{B}_J)_{k,k'} \leq 0$ for all 
    $k \neq k'$. Then there exists a constant $c < 0$ such that 
    $\mathbf{B}_J + c\, \Pi_s$ is stoquastic in the computational basis, where 
    $\Pi_s$ denotes the projector onto the symmetric subspace.
\end{thm}
\begin{proof}  
For the study of stoquasticity, we are only concerned with the off-diagonal elements in the computational basis.
Let us first consider the off-diagonal elements of $\mathbf{B}_J$ in the Dicke basis.
Since $\ket{D^k_{2J}}$ and $\ket{D^{k'}_{2J}}$ have different Hamming weights when $k \neq k'$, they have no computational basis states in common.
Consequently $\ket{D^{k'}_{2J}}\bra{D^{k}_{2J}}$ for $k \neq k'$ contribute only to off-diagonal elements in the computational basis, and nonpositive off-diagonal elements of $\mathbf{B}_J$ in the Dicke basis contribute only nonpositive off-diagonal elements in the computational basis.

The diagonal elements $(\mathbf{B}_J)_{k,k}$ may be positive, and the corresponding terms 
$(\mathbf{B}_J)_{k,k}\ket{D^k_{2J}}\bra{D^k_{2J}}$ contribute to both diagonal and 
off-diagonal elements in the computational basis, potentially resulting in positive 
off-diagonal elements.
Adding $c\,\Pi_s$ with $c < 0$ shifts all Dicke basis diagonal elements by $c$.
As $\Pi_s$ is diagonal in the Dicke basis, it does not affect the off-diagonal elements 
of $\mathbf{B}_J$ in the Dicke basis, leaving the first part of the argument intact.
For $c$ sufficiently negative, all diagonal contributions become nonpositive in the 
computational basis as well, completing the proof.
\end{proof}

In the Dicke basis, the corresponding state with the lowest eigenvalue is guaranteed to only have real and positive elements. 
The state corresponding to the maximal violation of Ineq. \eqref{eq:Ineq6} is even a Gaussian superposition of Dicke states 
$\ket{\psi_n}=\sum_{k=0}^n \psi_k^{(n)} \ket{D^k_n}$ where
\begin{equation}
\label{eq:GaussianDickeCoeff}
\psi_k^{(n)}:=\frac{e^{-(k-\mu)^2/4\sigma}}{\sqrt[4]{2\pi \sigma}},
\end{equation}
with $\mu :=n/2+A/(2B-C)$ and $e^{-1}/2\pi\ll \sigma \ll n$. 
In the next section, we investigate to which extent the results of this section are generalizable. 

\section{Stoquasticity in a generalized class of 2-body permutationally invariant Bell operators}\label{sec:GenTwoBodyClass}

The Bell operator of Eq. \eqref{eq:Ineq6Op} belongs to a class of permutationally invariant Bell operators
\begin{equation}
 \alpha \mathbf{S}_0 + \beta \mathbf{S}_1 + \frac{\gamma}{2}\mathbf{S}_{00} + \delta\mathbf{S}_{01}+ \frac{\varepsilon}{2}\mathbf{S}_{11},
 \label{eq:2BodyBellOp}
\end{equation}
with parameters $\alpha, \beta, \gamma, \delta, \varepsilon$ satisfying
\begin{equation}
\begin{aligned}
\alpha &= x[\sigma \mu + \tau (x+y)], \quad
\beta = \mu y, \\
\gamma &= x^2, \quad
\delta = \sigma x y, \quad
\varepsilon = y^2 ,
\end{aligned}
\label{eq:TwoBodyClassParam}
\end{equation}
where $\sigma,\tau \in \{-1,1\}$, and $x,y, \mu \in \mathbbm{N}$. 
Additionally, depending on whether the number of parties $n$ is even or odd, the parity of $\mu$ must be opposite to that of $x$ or $y$ respectively.
The operator of Eq. \eqref{eq:Ineq6Op} can be retrieved by choosing the parameters of the Bell operator to be $(\alpha, \beta, \gamma, \delta, \epsilon) = (-2,0,1,-1,1)$.
This occurs when the corresponding set of elementary parameters $(x,y,\sigma, \tau, \mu) =(1,1,-1,-1,0)$.

The class of Bell inequalities characterized by Eq. \eqref{eq:2BodyBellOp} and Eq. \eqref{eq:TwoBodyClassParam} is tangent to at least one vertex of the symmetrized 2-body local polytope~\cite{tura_nonlocality_2015}.
This is the 2-body local polytope projected into the subspace consisting of the probabilities satisfying $ P(\sigma(\vec{\alpha})|\sigma(\vec{x})) = P(\vec{\alpha}|\vec{x})$ with $\sigma \in \Pi_n$, where $\Pi_n$ is the symmetric group of $n$ elements.
The fact that the inequalities are tangent makes them particularly interesting for the detection of nonlocal correlations. 
To demonstrate that stoquasticity is not restricted to Eq. \eqref{eq:Ineq6Op}, but is instead a prevalent feature of this class, we prove the following theorem.

\begin{thm} \label{thm:TwoBodyGeneralClass}
For the class of Bell operators characterized by Eq. \eqref{eq:2BodyBellOp} and Eq. \eqref{eq:TwoBodyClassParam}, the Bell operator is stoquastic in every block under the conditions
\begin{align}
x \sin \varphi = -\sigma y \sin \theta,\label{eq:C_condition}
\end{align}
\begin{equation}
\begin{aligned}
 x \tau (x+y) \sin \varphi\leq 0.\label{eq:A_condition}
\end{aligned}
\end{equation}
\end{thm}

\begin{proof}

As the two-body symmetric Bell operator admits a block decomposition, showing that the full operator is stoquastic is equivalent to showing that each block $\mathbf{B}_J(\varphi, \theta)$ is stoquastic. 
The explicit expressions of elements on the first and second off-diagonal, $u_k$ and $v_k$ respectively, are 
\begin{equation}
\begin{aligned}
 u_k := [A' + (2J-1-2k)D]\sqrt{(2J-k)(k+1)}, \nonumber\\
 v_k := C\sqrt{(2J-k)(2J-k-1)(k+1)(k+2)}/2, \nonumber
\end{aligned}
\end{equation}
with
\begin{equation}
 A':=\alpha \sin \varphi + \beta \sin \theta, \nonumber
\end{equation}
\begin{equation}
 C:=\gamma \sin^2 \varphi + 2\delta \sin \varphi \sin \theta + \varepsilon\sin^2\theta, \nonumber
\end{equation}
\begin{equation}
\begin{aligned}
 D:=&\gamma \cos \varphi \sin \varphi + \delta \cos \varphi \sin \theta \\
 + &\delta
\cos \theta \sin \varphi + \varepsilon \cos \theta \sin \theta. \nonumber
\end{aligned}
\end{equation}
The index $k$ runs from $0$ to $2J$, so the square root terms are strictly nonnegative.
The constraints which need to be fulfilled such that each block $\mathbf{B}_J(\varphi, \theta)$ is stoquastic are thus the following:
\begin{align}
     C \leq 0, \nonumber\\
     A' + (2J-1-2k)D \leq 0. \nonumber
\end{align}
First, due to Eq. (\ref{eq:TwoBodyClassParam}), we have the following conditions: $\alpha, \delta \in \mathbb{Z}$ and $\beta, \gamma, \epsilon \in \mathbb{Z}_{\geq0}$.
We note that $\alpha$ and $\delta$ can be independently positive and negative. 
With this in mind, let us first consider the condition $C \leq 0$.
After substituting Eq. (\ref{eq:TwoBodyClassParam}) into it, we obtain the square $(x \sin \varphi + \sigma y \sin \theta)^2$. 
$C$ can thus never be negative. 
For it to be also nonpositive and therefore zero, the following condition must hold
\begin{equation}
    x \sin \varphi = -\sigma y \sin \theta,
    \nonumber
\end{equation}
which is Eq. \eqref{eq:C_condition}, the first condition of the theorem. 
Inserting the above condition into the expression for $D$, we obtain 
\begin{equation}
\begin{aligned}
    D = \bigg[ -\sigma x y \cos \varphi \sin \theta - \sigma x y \cos \theta \sin \varphi \\ 
    + \sigma x y (\cos \varphi \sin \theta + \cos \theta \sin \varphi) \big] \nonumber\\
    = xy \left(\cos \varphi \sin \theta + \cos \theta \sin \varphi \right) ( -\sigma + \sigma) = 0. \nonumber
\end{aligned}
\end{equation}
With this substitution, the condition $ A' + (2J-1-2k)D \leq 0$ thus reduces to $ A'\leq 0$.
Inserting Eq. \eqref{eq:TwoBodyClassParam} in $A' \leq 0$, we obtain $x(\sigma \mu + \tau (x+y)) \sin \varphi + \mu y \sin \theta \leq 0$. After also substituting the condition of Eq. \eqref{eq:C_condition} into it, we obtain 
\begin{equation}
\begin{aligned}
x(\sigma \mu + \tau (x+y)) \sin \varphi - \sigma \mu x \sin \varphi \\
= x \tau (x+y) \sin \varphi\leq 0,
\nonumber
\end{aligned}
\end{equation}
which is Eq. \eqref{eq:A_condition}, the second condition of the theorem. 
Thus if  Eq. \eqref{eq:C_condition} and Ineq. \eqref{eq:A_condition} simultaneously hold, the operator defined by Eq. \eqref{eq:2BodyBellOp} and Eq. \eqref{eq:TwoBodyClassParam} is stoquastic across all blocks. 

\end{proof}

Theorem \ref{thm:TwoBodyGeneralClass} shows that for all parameters defined by Eq. \eqref{eq:TwoBodyClassParam} characterizing the Bell operators of Eq. \eqref{eq:2BodyBellOp}, there exist measurement angles for which the Bell operator is stoquastic.
The conditions on the measurement parameters are rather stringent however. 
To see this, let us consider the following concrete example of Eq. \eqref{eq:Ineq6Op} which has parameters $(x,y,\sigma, \tau, \mu) =(1,1,-1,-1,0)$.
In this case, Eq. \eqref{eq:A_condition}
 and Eq. \eqref{eq:C_condition} become
\begin{equation}
     \sin \varphi =  \sin \theta, \; \; \;   - \sin \varphi \leq 0.
     \label{eq:Ineq6StoqCond}
\end{equation}
For $\varphi,\theta \in [0,2\pi)$,  Eq. \eqref{eq:Ineq6StoqCond} translates to $\varphi,\theta \in [0,\pi]$ and $(\varphi = \theta) \lor (\varphi = \pi-\theta)$.
Eq. \eqref{eq:Ineq6StoqCond} therefore imposes a stringent constraint on the relationship between the two measurement parameters $\varphi, \theta$. 

However, when optimizing a Bell operator over $\vec{\alpha}$ and $(\varphi,\theta)$ for the quantum violation, these stoquasticity requirements are generally not met. 
As the spectra of Bell operators are invariant under local unitary transformations $\mathbf{U}^{\otimes n}$, all operators of the form $\mathbf{U}^{\otimes n} \mathbf{B} (\mathbf{U}^{\dagger})^{\otimes n}$ yield the same maximum quantum violation. 
In particular, for the class of Bell operators defined by Eq.~\eqref{eq:2BodyBellOp} and Eq.~\eqref{eq:TwoBodyClassParam}, it was shown that under local orthogonal transformations of the form~\cite{tura_nonlocality_2015}
 \begin{equation}
 \label{eq:unitary}
  \mathbf{U}(\Delta):=\left(
  \begin{array}{rr}
  \cos(\Delta/2)&-\sin(\Delta/2)\\
  \sin(\Delta/2)& \cos(\Delta/2)
  \end{array}
  \right),
  \nonumber
 \end{equation}
the following holds $\mathbf{B}(\varphi+\Delta,\theta+\Delta) = \mathbf{U}(\Delta)^{\otimes n} \mathbf{B}(\varphi,\theta) (\mathbf{U}(\Delta)^{\dagger})^{\otimes n}$. 
This means that the spectrum only depends on the difference between the measurement parameters $\varphi$ and $\theta$.
Additionally, for a fixed $\Delta$ and otherwise arbitrary values of $\varphi$ and $\theta$, the operator is generally not stoquastic and the corresponding ground state is generally not a positive linear combination of Dicke states~\cite{tura_nonlocality_2015}. 

Though for the class of Bell operators defined by Eq.~\eqref{eq:2BodyBellOp} and Eq.~\eqref{eq:TwoBodyClassParam}, stoquasticity reduces to a single-parameter problem due to the shared measurement settings, finding local unitaries that render a general Bell operator stoquastic is NP-complete~\cite{marvian_computational_2019}.
Rather than approaching stoquasticity through constraints on the measurement parameters, we develop a framework for characterizing the full region of stoquastic Bell coefficients at fixed measurement parameters.

\section{Stoquasticity Cone}\label{sec:StoqCone}

A full characterization of the stoquastic parameter regimes for permutationally invariant Bell operators follows from expressing the off-diagonal constraints as linear inequalities on the Bell coefficients. 
Any permutationally invariant Bell operator $\mathbf{B}$ can be written as a linear combination of PI measurement operators $\mathbf{B}=\sum_i \alpha_i \mathbf{S}_{i}$, where the label $i$ encodes both the order of the symmetrized operator and its measurement settings, cf. Sec. \ref{Sec:prelim}. 
The matrix elements of $\mathbf{B}$ implicitly depend on $\varphi,\theta$ through trigonometric relations, since elements of the symmetric measurement operators $\mathbf{S}_i$ depend trigonometrically on $\varphi,\theta$, cf. App. \ref{app:HighToSingle}. 
Finding closed-form conditions on $\alpha_i$, $\varphi$, and $\theta$ that make the Bell operator stoquastic is therefore nontrivial.

However, for fixed measurement parameters $\varphi,\theta$, the stoquasticity conditions on the off-diagonals of the Bell operator, $\mathbf{B}_{k,l} \leq 0 \; \forall k \neq l$, give rise to a set of linear inequalities on the Bell coefficients of the form $\sum_{i}  (\mathbf{S}_{i})_{k,l} \alpha_i \leq 0 \; \forall k \neq l$.
Here $(\mathbf{S}_{i})_{k,l} \in \mathbbm{R}$ is the $k,l$-th matrix element of the $i$-th symmetrized measurement operator. 
We note that only matrix elements up to the $K$-th off-diagonal, where $K$ is the order of the Bell operator, give rise to nontrivial constraints as the matrix elements on higher off-diagonals are $0$, cf. Sec. \ref{sec:PIBlock}.
The constraints on the $\alpha_i$ define hyperplanes.
Moreover, since the system of inequalities is homogeneous, the feasible region of the $\alpha_i$'s forms a cone, which is a special case of a polyhedron.
We denote this cone as follows
\begin{equation}
\mathbbm{S}^{n,K}_{\varphi,\theta} = \{\vec{\alpha} \in \mathbbm{R}^m | (\vec{\mathbf{S}}_{k,l})^T \cdot \vec{\alpha} \leq 0 \; \forall k\neq l\},
\label{eq:PolPrimal}
\end{equation}
where 
\begin{equation}
\begin{aligned}
(\vec{\mathbf{S}}_{k,l})^T = \big( &(\mathbf{S}_0)_{k,l}, (\mathbf{S}_1)_{k,l}, (\mathbf{S}_{00})_{k,l}, 
(\mathbf{S}_{01})_{k,l}, (\mathbf{S}_{11})_{k,l},\\
&...,(\mathbf{S}_{0\mathbf{1}_{K-1}})_{k,l}, (\mathbf{S}_{\mathbf{1}_{K}})_{k,l}
\big) \nonumber
\end{aligned}
\end{equation}
where ${\mathbf{1}_{K}}$ denotes the string of $K$ ones and $\vec{\alpha}=(\alpha_0,\ldots,\alpha_i,\ldots,\alpha_m)$ with $m$ being the number of measurement settings.
The variables $n, K$ denote the number of parties and the maximum on the order of the PI measurement operators respectively. 

Though in the previous sections we focused on the two-input-two-output scenario, the construction for the cone extends to Bell operators with an arbitrary number of inputs and outputs.
As the cone $\mathbbm{S}^{n,K}_{\varphi,\theta}$ is by definition the set of $\vec{\alpha}$ that yield stoquastic Bell operators for a set of fixed measurement parameters, solving $\mathbbm{S}^{n,K}_{\varphi,\theta}$ therefore also provides a complete characterization of all Bell coefficients that render the corresponding operator stoquastic.

To achieve such a characterization, we turn to the dual description of $\mathbbm{S}^{n,K}_{\varphi,\theta}$.
Rather than characterizing the permissible $\vec{\alpha}$ indirectly through the hyperplane conditions, the dual description directly expresses any feasible Bell coefficient as a (positive) linear combination of the dual generators of the cone.
We first note that as Eq. \eqref{eq:PolPrimal} is homogeneous, the origin is always a feasible point, i.e. $\vec{0} \in \mathbbm{S}^{n,K}_{\varphi,\theta}$.
Additionally, $\mathbbm{S}^{n,K}_{\varphi,\theta}$ is generally unbounded and it may contain rays that extend in both directions, which are simply denoted lines.
Therefore, $\mathbbm{S}^{n,K}_{\varphi,\theta}$ is not pointed. 
The dual description
\begin{equation}
\begin{aligned}
\mathbbm{S}^{n,K}_{\varphi,\theta} = \bigg\{\vec{\alpha} \in \mathbbm{R}^m |
\vec{\alpha} = \sum_i &\lambda_i \vec{r}_i + \sum_j \mu_j \vec{l}_j,\\
\; &\lambda_i \geq 0, \; \mu_j \in \mathbbm{R} \bigg\}, 
\label{eq:PolDual}
\end{aligned}
\end{equation}
can thus be fully characterized by a positive combination of the extreme rays, $\vec{r}_i$, and a linear combination of the lines, $\vec{l}_j$, i.e. no vertices are needed. 

Lines are particularly interesting to consider: they provide extra degrees of freedom that do not alter the off-diagonal elements and thus the stoquasticity of the resulting Bell operator.
However, note that different linear combinations do change the Bell inequality, thus yielding richer expressivity. 
This is due to the fact that for a line $\vec{l}_j$
\begin{equation}
\sum_i (\mathbf{S}_i)_{k,l} (\vec{l}_j)_i \le 0 ,\;\forall, \; k\neq l \nonumber
\end{equation}
and
\begin{equation}
\sum_i -(\mathbf{S}_i)_{k,l} (\vec{l}_j)_i \le 0 ,\; \forall, \; k\neq l \nonumber
\end{equation}
hold simultaneously. 
Here $(\vec{l}_j)_i$ denotes the $i-$th element of $\vec{l}_j$.
However, the above is only possible if $ \sum_{i} (\mathbf{S}_{i})_{k,l} (\vec{l}_j)_i = 0 \; \forall \; k \neq l$.
Any linear combination of lines thus does not change the off-diagonal elements of the Bell operator, i.e. for
\begin{equation}
    \vec{\alpha}_1 = \sum_i \lambda_i \vec{r}_i + \sum_j \mu^1_j \vec{l}_j, \nonumber
\end{equation}
and
\begin{equation}
    \quad \vec{\alpha}_2 = \sum_i \lambda_i \vec{r}_i + \sum_j \mu^2_j \vec{l}_j, \nonumber
\end{equation}
$\mathbf{B}_{k,l}(\vec{\alpha}_1) = \mathbf{B}_{k,l}(\vec{\alpha}_2) \; \forall \; k \neq l$ and $\mu^i_j \in \mathbbm{R}$, where $\mathbf{B}_{k,l}(\vec{\alpha}_x)$ denotes the $(k,l)$-th element of the Bell operator with Bell coefficients $\vec{\alpha}_x$.
For a fixed set of $\lambda_i$'s, there is thus an affine subspace over the parameters $\mu_j$ which leaves all the off-diagonals invariant.

Quantities such as the classical bound $\beta_C$ and the quantum bound $\beta_Q$, however are generally different for different values of $\vec{\alpha}_x$.
Here $\beta_Q$ is the lowest expected value of a Bell operator optimized over all quantum states and $\beta_C$ is the lowest expected value of the same Bell operator optimized over all local deterministic strategies.
A key figure of merit for Bell operators is the quantum-classical gap  $\beta_Q/\beta_C$. 
Since changing the coefficients $\mu_i$ of the lines modifies $\vec{\alpha}$ while leaving the off-diagonal invariant,
this freedom can be used to optimize the Bell coefficients to maximize the gap between the quantum and classical bounds, without affecting the stoquasticity. 

\subsection{Two-body scenario}\label{sec:2BodyRes}

Two-body permutationally invariant Bell operators provide the simplest setting in which the stoquasticity constraints can be analyzed, while still allowing for a separation between the best quantum value $\beta_Q$ and the best classical value $\beta_C$.
Moreover, although the number of hyperplanes grows with $n$, the number of irredundant hyperplanes is constant.
Additionally, the number of extremal lines and rays that characterize the dual description is independent of $n$ as well, yielding a compact parametrization of all admissible Bell coefficients. Such a parametrization enables efficient numerical optimization of the quantum–classical gap within the class of stoquastic operators.
As noted in Section \ref{sec:PIBlock}, we focus on the symmetric subspace, which corresponds to focusing on the block $\mathbf{B}_{n/2}$.

For the characterization of the hyperplane description of $\mathbbm{S}^{n,2}_{\varphi,\theta}$ in terms of only the irredundant hyperplanes, we calculate the analytical expression for each off-diagonal element of the two-body Bell operator, cf. App. \ref{app:HighToSingle} and \ref{app:TwoBodyHPlanes}.
Then using the fact that all other hyperplanes can be written as conical combinations of three hyperplanes, we show that there are only three irredundant ones. 
In particular, the form of the irredundant hyperplanes is as follows
\begin{align}
\label{eq:TwoBodyHyperplanes}
         \alpha_0 s^1_0 + \alpha_1 s^1_1 + \alpha_2 s^1_2+ \alpha_3 s^1_3+ \alpha_4 s^1_4 &\leq 0, \nonumber\\
     \alpha_0 s^1_0 + \alpha_1 s^1_1 - \alpha_2 s^1_2 - \alpha_3 s^1_3 - \alpha_4 s^1_4 &\leq 0, \; \\
      \alpha_2 s^2_2 + \alpha_3 s^2_3 + \alpha_4 s^2_4 &\leq 0, \; \nonumber
\end{align}
where the $s^i_j \in \mathbb{R}$.
See App. \ref{app:TwoBodyHPlanes} for the exact form in terms of $\varphi,\theta$.

In the first two conditions of Eq. \eqref{eq:TwoBodyHyperplanes}, the coefficients for the first two variables $\alpha_0, \alpha_1$ are the same. 
For the remaining three variables $\alpha_2, \alpha_3, \alpha_4$, they differ by a minus sign. 
We thus note that there is always a valid solution for all measurement parameters by taking $(\alpha_2, \alpha_3, \alpha_4) = (0,0,0)$ and $\alpha_0 \leq \frac{s^1_1}{s^1_0} \alpha_1$.
However, we note that the resulting Bell operator has limited utility. 
It only consists of single-body terms and does not allow for a gap between $\beta_Q$ and $\beta_C$. 

To obtain an explicit parametrization of all permissible stoquastic Bell coefficients $\vec{\alpha} \in \mathbbm{S}^{n,2}_{\varphi,\theta}$, we move to the dual description and solve it analytically, cf. App. \ref{app:TwoBodyPolyhedron}.
In there, we provide explicit expressions for the rays $r_i$ and lines $l_j$ of the cone $\mathbbm{S}^{n,2}_{\varphi,\theta}$ in terms of $\varphi, \theta, n$.
The most interesting result is that the cone $\mathbbm{S}^{n,2}_{\varphi,\theta}$  is generally characterized by three rays and two lines. 
Therefore, each point $\vec{\alpha} \in \mathbbm{S}^{n,2}_{\varphi,\theta}$ can be written as a weighted combination of 5 variables
\begin{equation}
    \label{eq:TwoBodyRaysAndLines}
    \vec{\alpha} = c_1 \vec{r_1} + c_2 \vec{r_2} + c_3 \vec{r_3} + c_4 \vec{l_1} + c_5 \vec{l_2},
\end{equation}
with $c_1,c_2,c_3 \in \mathbbm{R}_{\geq 0}$ and $c_4, c_5 \in \mathbbm{R}$. 
The small number of variables required to parameterize the Bell coefficients $\alpha$ makes it suitable for numerical optimization. 

To demonstrate this utility,  we optimize the two-body permutationally invariant Bell operator with respect to the quantum-classical gap $\beta_Q/\beta_C$ at $(\varphi, \theta) = (\pi/6, 5\pi/6)$, which are the optimal measurement parameters for Ineq. \eqref{eq:Ineq6}.
The expression of the rays and lines for these measurement parameters are
\begin{equation}
\label{eq:I6RaysAndLines}
\begin{aligned}
    \vec{r_1}&=(-\sqrt{3} (n-1),0,1,-1,0),\\
    \vec{r_2}&=(-\sqrt{3} (n-1),0,-1,1,0),\\
    \vec{r_3}&=(0,0,0,-1,0),\\
    \vec{l_1}&=(-1,1,0,0,0),\\
    \vec{l_2}&=(0,0,1,-2,1).
\end{aligned}
\end{equation}
Numerically optimizing $\vec{c}$ for the quantum-classical gap $\beta_Q/\beta_C$, we obtain the coefficients listed in Table \ref{tb:OptGap}, cf. App. \ref{app:TwoBodyConeOpt}.
The optimized stoquastic Bell operators yield values of $\beta_Q/\beta_C$ that closely match those of the Bell operator of Eq. \eqref{eq:Ineq6Op}.
This suggest that the Bell operator of Eq. \eqref{eq:Ineq6Op} is optimal with respect to stoquasticity. 
Additionally, the ground states of the optimized Bell operator across all values of $n$ are Gaussian-like, as is displayed in Fig. \ref{fig:OptGS}. 

\subsection{Three-body scenario}\label{sec:3BodyRes}

Three-body permutationally invariant Bell operators constitute the next level of complexity, where stoquasticity constraints remain tractable while the structure of the admissible coefficient space becomes substantially richer.
As every two-body permutationally invariant Bell operator admits a choice of Bell coefficients and measurement parameters that renders it stoquastic in the symmetric subspace, it is natural to ask whether analogous statements can be found for three-body operators.
To this end, let us first define the three-body permutationally invariant Bell operator as
\begin{equation}
\begin{aligned}
 \label{eq:ThreeBodyBellOp}
 \mathbf{B}_3 \coloneqq \bigg[ &\alpha_0  \mathbf{S}_0 + \alpha_1  \mathbf{S}_1 + \alpha_2 \mathbf{S}_{00} + \alpha_3  \mathbf{S}_{01} + \alpha_4 \mathbf{S}_{11} +  \\ 
 & \alpha_5 \mathbf{S}_{000} +  \alpha_6 \mathbf{S}_{001} +  \alpha_7 \mathbf{S}_{011}+  \alpha_8 \mathbf{S}_{111} \bigg] . 
 \nonumber
 \end{aligned}
\end{equation}
For the three-body scenario, the number of non-redundant hyperplanes does not collapse to three, but generally depends on the number of parties $n$.
This complicates the analytical treatment. 
Nonetheless, we can show that the Bell operator can be made stoquastic for all measurement parameters, as is captured by the following theorem. 
\begin{thm}
    The three-body permutationally invariant Bell operator $\mathbf{B}_3$ can be made stoquastic in the symmetric block for all $\varphi,\theta$. 
\end{thm}
\begin{proof}
To prove that for all $\varphi,\theta$ there exists a combination of Bell coefficients $\vec{\alpha}$ such that the corresponding three-body Bell operator $\mathbf{B}_3$ is stoquastic, it is sufficient to prove that for $\vec{\alpha}=(\alpha_0,0,\alpha_2,0,0,\alpha_5,0,0,0)$ the Bell operator $\mathbf{B}_3$ is stoquastic.
As $\alpha_0$ corresponds to $\mathbf{S}_0$, $\alpha_2$ corresponds to $\mathbf{S}_{00}$ and $\alpha_5$ corresponds to $\mathbf{S}_{000}$, a PI measurement operator of each order is included. However, notice that through this choice of $\vec{\alpha}$ we have eliminated any dependence on $\theta$. The explicit expressions for the matrix elements of the three-body PI measurement operator can be found in App. \ref{app:HighToSingle}.

Since only three-body PI measurement operators have nonzero elements on the third 
off-diagonal, the stoquasticity condition on the third off-diagonal reduces to
\begin{align*}
     \sin ^3(\varphi) \alpha_5 \leq 0. 
\end{align*}
which can always be satisfied by an appropriate choice of $\alpha_5$.
For fixed $\alpha_5$, the stoquasticity condition on the second off-diagonal is
\begin{align*}
 \sin ^2(\varphi ) \alpha_2 + (n -2 k-2) \sin ^2(\varphi ) \cos (\varphi ) \alpha_5 \leq 0,
\end{align*}
where $k \in \{0,1,\ldots,n-2\}$ indexes the matrix elements on the second off-diagonal.
This is a finite linear inequality in $\alpha_2$, and can therefore always be satisfied 
by choosing $\alpha_2$ sufficiently negative.
For fixed $\alpha_2$ and $\alpha_5$, the stoquasticity condition on the first 
off-diagonal is
\begin{equation}
\begin{aligned}
\sin (\varphi ) \alpha_0 
+  (-2 k+n-1) \sin (2 \varphi )   \alpha_2 \\
+ \bigg[ 3 \sin (\varphi ) \left(4 k^2-4 k (n-1)+n^2-3 n+2\right) \cos ^2(\varphi ) \\
+k (-k+n-1) \sin ^2(\varphi) \bigg] \alpha_5 \leq 0. 
\end{aligned}
\label{eq:threeBodyFirstDiag}
\nonumber
\end{equation}
which is likewise a finite linear inequality in $\alpha_0$, and can always be satisfied 
by choosing $\alpha_0$ such that $\textrm{sgn}(\alpha_{0})=-\textrm{sgn}(\sin (\varphi))$ and $|\alpha_{0}|$ sufficiently large.
Since all other off-diagonals vanish for our choice of $\vec{\alpha}$, the 
stoquasticity conditions can be simultaneously satisfied for any $\varphi, \theta$.
\end{proof}
As the number of hyperplanes increases with $n$, the number of rays generally increases with $n$ as well. 
The complexity of the characterization of $\mathbbm{S}^{n,3}_{\varphi,\theta}$ therefore grows with $n$.
At the measurement parameters $(\varphi,\theta) = (\pi/6,5\pi/6)$ and values $n \in \{10,20,30,40,50\}$ for example, we observe that the total number of extremal rays is $\{13,23,33,43,53\}$.
However, for other values of measurement parameters, this is generally not the case, as in addition to $n$, the total number of rays depends on $(\varphi,\theta)$.
Analytically, we find the following expressions for the lines
\begin{align*}
    \vec{l}_1 = (-\csc(\phi) \sin(\theta) , 1 , 0 , 0 , 0, 0,0,0,0), 
\end{align*}
\begin{align*}
        \vec{l}_2 = (0,0,\sin ^2(\theta ) \csc ^2(\phi ),-2 \sin (\theta ) \csc (\phi ),1, 0,0,0,0), 
\end{align*}
\begin{equation}
\begin{aligned}
    \vec{l}_3 = (0,0,0,0,0,&\sin ^3(\theta ) \left(-\csc ^3(\phi )\right),\\
    &3 \sin ^2(\theta ) \csc ^2(\phi ),-3 \sin (\theta ) \csc (\phi ),1), \nonumber
\end{aligned}
\end{equation}
cf. App. \ref{app:ThreeBodyLines} for details on how they are obtained. 
Numerical evidence suggests that up to $n=50$, there are indeed only three lines and that the number is independent of $n, \varphi, \theta$. 

We notice that there is a splitting in the lines, where $\vec{l}_1$ is only non-zero at the single-body entries, $\vec{l}_2$ is only non-zero at the two-body entries and $\vec{l}_3$ is only non-zero at the three-body entries.
This suggests that terms in the Bell operators associated with PI measurement operators of different orders $K$ cannot generally cancel each other. 
This is expected as terms corresponding to different orders scale differently in $n,k$, cf. App. \ref{app:HighToSingle} for the analytical expressions.

Similarly to the two-body scenario, we now optimize over the three-body cone. 
Our numerical optimization suggests that at $(\varphi,\theta) = (\pi/6, 5\pi/6)$, the quantum-classical gap of Eq. \eqref{eq:Ineq6Op} is already optimal with respect to stoquasticity.
Additionally, for the measurement parameters $(\varphi,\theta) = (\pi/4, -\pi/4)$, our optimization suggest that the operator of the form 
\begin{equation}
    \mathbf{B} = -2 \mathbf{S}_0 + \frac{1}{2} \mathbf{S}_{00} +  \mathbf{S}_{01} +\frac{1}{2} \mathbf{S}_{11},
    \label{eq:pi4Op}
\end{equation}
is optimal. 
For $n=10,20,30,50$, the corresponding values for the quantum-classical gaps are $\beta_{Q}/\beta_{C}=1.02904, 1.06576, 1.08556, 1.10791$ respectively. 
We note the remarkable similarity between Eq. \eqref{eq:pi4Op} and Eq. \eqref{eq:Ineq6Op}, with the difference being a factor $-1$ in the $\mathbf{S}_{01}$ term. 
The former can be transformed into the latter without affecting the quantum-classical gap by replacing $\mathbf{S}_{01}$ by $ -\mathbf{S}_{01}$.
This is equivalent to relabeling $1 \rightarrow -1$ and $-1 \rightarrow 1$ for the outcomes of $\mathbf{M}_1$. 
As $\mathbf{S}_{11}$ contains pairs of $\mathbf{M}_1$ operators, the sign changes cancel each other and the operator remains the same. 

In Fig. \ref{fig:I6Scatter}, we have plotted the measurement parameters $(\varphi, \theta)$ for which Bell operators of the form Eq. \eqref{eq:Ineq6Op} produce quantum-classical gaps $\beta_{Q}/\beta_{C} \geq 1$, where $\beta_Q$ denotes the minimal eigenvalue of the Bell operator and $\beta_C$ denotes the minimum classical value.
Remarkably, the area of the region with such violations is dominant and as it increases with $n$, suggesting that Eq. \eqref{eq:Ineq6Op} is a good candidate for providing stoquastic Bell violations for a substantial range of the measurement parameters. 
Additionally, we remark that the optimized Bell operators do not contain three-body terms, suggesting that two-body operators suffice to achieve the optimal quantum-classical gap within the stoquastic cone at the measurement parameters considered. 
Whether three-body or higher-order terms can yield strictly larger gaps remains an open question.

\section{Higher order operators}\label{sec:HighOrderDistr}

So far, the study of PI Bell operators has primarily focused on ground states that are Gaussian-like or of the spin-squeezed form~\cite{tura_nonlocality_2015,schmied_bell_2016,engelsen_bell_2017}.
As noted before, the ground state $\ket{\phi}$ of a stoquastic Hamiltonian has only nonnegative elements and the square of its elements can therefore be interpreted as a probability distribution. 
In this section, we will show that any state with arbitrary nonnegative amplitudes in the Dicke basis can be realized as the ground state of a suitable linear combination of PI measurement operators restricted to the symmetric subspace $\mathcal{H}_{\text{sym}}$. 
This means that arbitrary discrete probability distributions of $n+1$ values can be sampled by measuring in the Dicke basis $\{ \ket{D_n^k} \}_{k=0}^n$.

To see this, let us consider the state
\begin{align}
    \ket{\phi} = \sum_k \phi_k \ket{D_n^k},
    \label{eq:stoqGS}
\end{align}
with $\phi_k \in \mathbbm{R}_{\geq 0}$ and $\sum_{k=0}^n \phi_k^2 =1$.
Then the following operator 
\begin{equation}
    H_{\phi} \coloneqq \mathbbm{I} - \ket{\phi}\bra{\phi}
    \label{eq:StoqParHam}
\end{equation}
is a symmetric and stoquastic matrix in the Dicke basis. 
We note that by construction $H_{\phi}$ is a parent Hamiltonian of $\ket{\phi}$, i.e. $\ket{\phi}$ is its ground state, albeit its expression may potentially involve up to $n$-partite operators.
By construction, $H_{\phi}$ is defined on the symmetric subspace, however it can be extended to all blocks by extending its support to the full space. 
As the ground state remains unchanged after this extension, we can w.l.o.g. focus on the symmetric subspace for the rest of the argument.
To see that any such $H_{\phi}$ can be written as a linear combination of PI measurement operators, let us first note that the $K$-body PI Pauli operators 
\begin{equation}
    \mathbf{T}_{\vec{w}(K)}
\coloneqq \sum_{\substack{\vec{i}=(i_1,\dots,i_n) \in \mathbbm{Z}_4^{\times n}\\
\vec{w}(K) = (w_1(\vec{i}), w_2(\vec{i}), w_3(\vec{i}) )} }
\mathbf{P}(\vec{i}),
\end{equation}
when projected onto the symmetric space, forms an overcomplete set of generators for the space of symmetric Hamiltonians.
Here $\mathbf{P}(\vec{i}) := P(i_1) \otimes P(i_2) \otimes \cdots \otimes P(i_n)$ represents the Pauli string acting on $n$ qubits, $P(0),P(1),P(2),P(3)= \mathbbm{I},X,Y,Z$ respectively, $\vec{w}(K)$ a vector of three nonnegative integers that sums to $K$ and $w_1(\vec{i})$ counts the number of $1$'s in $\vec{i}$, $w_2(\vec{i})$ the number of $2$'s and $w_3(\vec{i})$ the number of $3$'s. 
For example the operator $\mathbf{S}_{XXZY}$ corresponds to $\mathbf{T}_{(2,1,1)}$ and  $\mathbf{S}_{XY}$ corresponds to $\mathbf{T}_{(1,1,0)}$. 

We will now show that the set $\{\mathbf{T}_{\vec{w}(K)} \big|_{\text{sym}} \}$ over all possible vectors $\vec{w}(K)$ over all $K$ such that $1 \leq K \leq n$ forms a generating set for the operators supported on the symmetric subspace $\mathcal{H}_{\text{sym}}$, where $\mathbf{T}_{\vec{w}(K)} \big|_{\text{sym}}$ denotes the restriction of $\mathbf{T}_{\vec{w}(K)}$ to $\mathcal{H}_{\text{sym}}$. 
Consider the matrix $\mathbf{A}$ which has support only on $\mathcal{H}_{\text{sym}}$.
Any Hermitian matrix $\mathbf{A}$ can be written as a real linear combination of Pauli strings of length $n$, i.e. $\mathbf{A} = \sum_{\vec{i}} c_{\vec{i}} \mathbf{P}(\vec{i})$, since the set of Pauli strings of length $n$ forms a basis for the real vector space of Hermitian operators acting on $n$ qubits.
As $\mathbf{A}$ only has support on $\mathcal{H}_{\text{sym}}$, it satisfies $\mathbf{A} = \frac{1}{n!}\sum_{\sigma \in \Pi_n} \mathbf{U}_\sigma \mathbf{A} \mathbf{U}^\dagger_\sigma$, where $\Pi_n$ is the symmetric group over $n$ elements. 
Therefore 
\begin{equation}
\begin{aligned}
\mathbf{A} =& \frac{1}{n!}\sum_{\sigma \in \Pi_n} \mathbf{U}_\sigma \sum_{\vec{i}} c_{\vec{i}} \mathbf{P}(\vec{i}) \mathbf{U}^\dagger_\sigma \nonumber \\
= & \frac{1}{n!} \sum_{\vec{i}} c_{\vec{i}} \mathbf{T}_{\vec{w}(K)}, \nonumber
\end{aligned}
\end{equation}
where from the second to the third line, we have used that $ \sum_{\sigma \in \Pi_n} \mathbf{U}_\sigma \mathbf{P}(\vec{i}) \mathbf{U}^\dagger_\sigma = \mathbf{T}_{\vec{w}(K)}$.
As $\mathbf{A} = \mathbf{A} \big|_{\text{sym}}$ by construction, we have $\mathbf{A} = \frac{1}{n!} \sum_{\vec{i}} c_{\vec{i}} \mathbf{T}_{\vec{w}(K)}\big|_{\text{sym}}$.
This means that any operator with only nonzero support in the symmetric subspace $\mathcal{H}_{sym}$ can be written as a linear combination of $K$-body PI Pauli strings restricted to $\mathcal{H}_{sym}$.
This means that for every state $\ket{\phi}$ of the form Eq. \eqref{eq:stoqGS}, there is an associated parent Hamiltonian of the form of Eq.
\begin{equation}
     H_\phi = \sum_K \sum_{\vec{w}(K)} \gamma_{\vec{w}(K)} \mathbf{T}_{\vec{w}(K)},
     \label{eq:ParentHamDecomp}
\end{equation}
where $\gamma_{\vec{w}(K)} \in \mathbbm{R}$ are the weights of the linear combination.

We note, however, that $K$ will be determined by the Hamming weight of those $\vec{i}$ such that $c_{\vec{i}}\neq 0$, which can go up to $n$ in general.
To generate arbitrary probabilities, PI measurement operators of order $K \sim \mathcal{O}(n)$ are thus generally needed. 
While for parent Hamiltonians of some Gaussian-like states it is sufficient to consider PI operators of order two, cf. Eq. \eqref{eq:GaussianDickeCoeff}, generally this is not the case as we will illustrate with the n-qubit GHZ state 
\begin{equation}
    \ket{\text{GHZ}} = \frac{\ket{0}^{\otimes n} + \ket{1}^{\otimes n}}{\sqrt{2}}.
\end{equation}
In the symmetric Dicke basis, the parent Hamiltonian is of the form 
\begin{equation}
    \mathbbm{I} - \frac{\ket{D^0_n} \bra{D^0_n} + \ket{D^0_n} \bra{D^n_n} + \ket{D^n_n} \bra{D^0_n} + \ket{D^n_n} \bra{D^n_n}}{2}.
\end{equation}
In addition to diagonal matrix elements, it also has matrix elements at the top right corner $(k,l) = (0,n)$ and the bottom left corner $(k,l) = (n,0)$. 
This can only be achieved through a full body PI Pauli operator ($K=n$), since the term $\ket{D_{n}^{0}}\!\bra{D_{n}^{n}}=\ket{00\ldots0}\!\bra{11\ldots1}$ only appears in the operators $T_{\vec{w}(K)}$ with $K=n$. 
Generally, for arbitrary probability distributions over $k$, the corresponding state $\ket{\phi}$ requires $H_{\phi}$ to contain $\mathbf{T}_{\vec{w}(K)}$ of arbitrary order. 
The exploration of such parent Hamiltonians containing higher order operators is left for future work. 

\section{Conclusion and outlook}\label{sec:Conclusion}

In this work, we have established the first systematic connection between PI Bell operators and stoquasticity in the binary-input and binary-output scenario. As a motivating example, we showed that the Bell operator used in the largest many-body Bell experiments is stoquastic at its optimal measurement parameters. We then showed that stoquasticity is a general feature of a class of tangent two-body PI Bell inequalities, deriving explicit conditions on the measurement parameters under which the entire class is stoquastic. To go beyond this class, we introduced the stoquasticity cone, which provides a full characterization of all stoquastic PI Bell operators via its extremal rays and lines. Using this framework, we showed that all two-body and three-body PI Bell operators can be made stoquastic at any measurement parameters, and that the Bell operator of the motivating example is optimal with respect to stoquasticity.

Our results suggest that stoquasticity is not a fine-tuned property of PI Bell operators, but rather a robust feature that can always be enforced through an appropriate choice of Bell coefficients. The fact that the experimentally relevant Bell operator of Refs.~\cite{schmied_bell_2016,engelsen_bell_2017} is optimal with respect to stoquasticity at its natural measurement parameters suggests a deeper connection between experimental accessibility and stoquasticity that warrants further investigation. Furthermore, the parent Hamiltonian construction reveals that stoquastic PI Bell operators are surprisingly expressive.
While Gaussian-like ground states, which are the ones naturally arising in current experiments \cite{schmied_bell_2016, engelsen_bell_2017}, require only two-body correlators, accessing the full range of probability distributions requires higher-order operators of order $K \sim \mathcal{O}(n)$.

Several directions remain open for future investigation. A full analytical characterization of the extremal rays of the three-body stoquasticity cone for general system sizes remains an open problem. The framework of the stoquasticity cone is not specific to the PI setting and could in principle be extended to more general Bell scenarios, opening up the possibility of connecting stoquasticity to nonlocality beyond the permutationally invariant case. On the side of higher-order operators, developing an efficient framework for systematically investigating PI Bell operators beyond the three-body level is an important next step, as the complexity of PI measurement operators grows quickly with the operator order $K$. More broadly, characterizing how the geometry of the stoquasticity cone evolves with the order $K$ would reveal how the accessible stoquastic Bell correlations expand as higher-order terms are included.

\acknowledgements
J.L. and J.T. acknowledge the support received by the Dutch National Growth Fund
(NGF), as part of the Quantum Delta NL programme. 
J.T. acknowledges the support received from the European Union's Horizon Europe research and innovation programme through the ERC StG FINE-TEA-SQUAD (Grant No.~101040729). 
This publication is part of the `Quantum Inspire - the Dutch Quantum Computer in the Cloud' project (with project number [NWA.1292.19.194]) of the NWA research program `Research on Routes by Consortia (ORC)', which is funded by the Netherlands Organisation for Scientific Research (NWO).
This work was supported by the Netherlands Organization for Scientific Research (NWO/OCW), as part of Quantum Limits (project number SUMMIT.1.1016).

The views and opinions expressed here are solely
those of the authors and do not necessarily reflect those of the funding institutions. 
None of the funding institutions can be held responsible for them.

\begin{figure*}[h]
    \centering
    \includegraphics[width=0.4\textwidth]{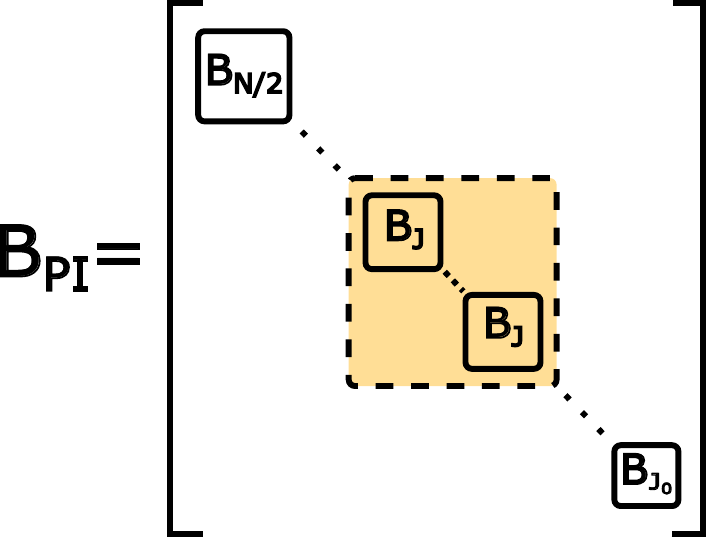}
    \caption{Schematic representation of the block structure of the permutationally invariant operator $B_{\mathrm{PI}}$.
    Here $J \in \{J_0, J_0 + 1, \ldots, n/2\}$ with $J_0 =0$ if $n$ is even and $J_0 =1/2$ if $n$ is odd}
    \label{fig:rhoPI_block}
\end{figure*}
\begin{figure*}[h]
\centering

\subfloat[$n=10$]{%
  \includegraphics[width=0.24\textwidth]{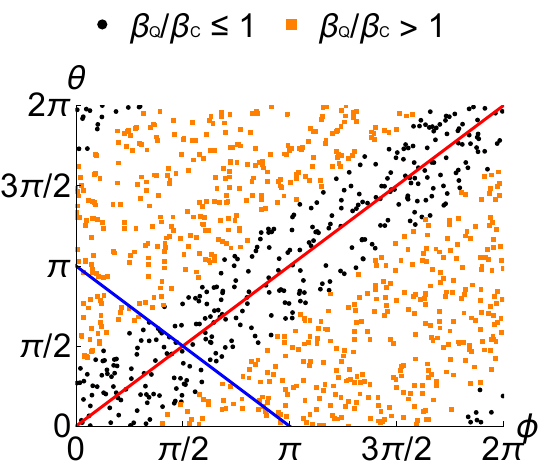}
}
\hfill
\subfloat[$n=20$]{%
  \includegraphics[width=0.24\textwidth]{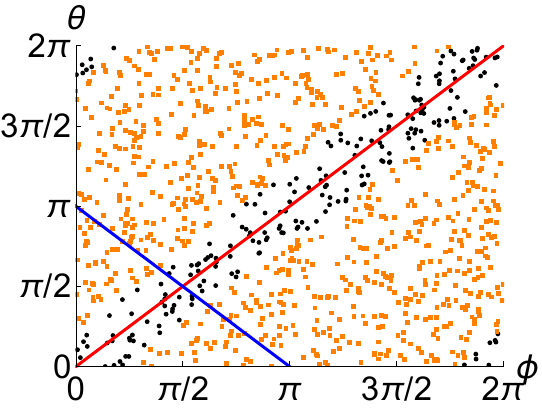}
}
\hfill
\subfloat[$n=30$]{%
  \includegraphics[width=0.24\textwidth]{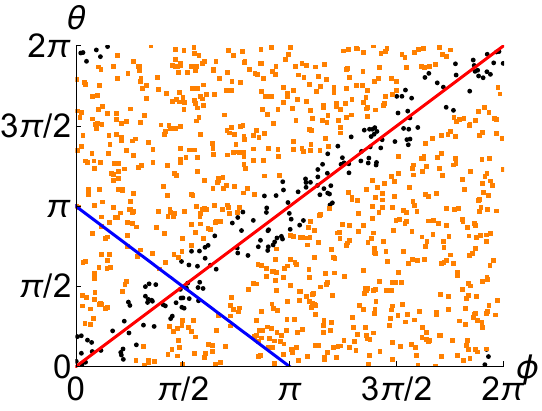}
}
\hfill
\subfloat[$n=50$]{%
  \includegraphics[width=0.24\textwidth]{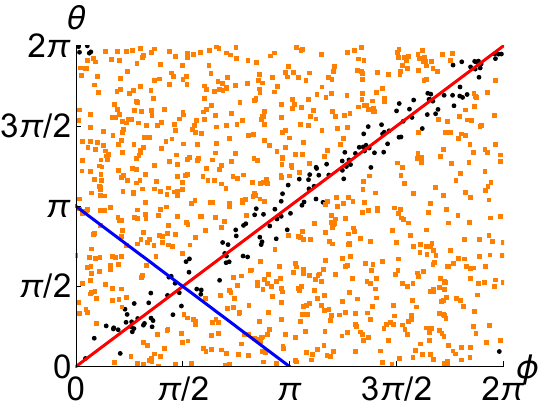}
}
\caption{The orange points denote values of measurement parameters $(\varphi,\theta)$ with a quantum violation above 1 and the black points denote measurement parameters without such a violation. 
The blue line denotes $\pi-\varphi = \theta$ and the red line denotes $\varphi = \theta$, which correspond to values that satisfy the stoquastic conditions of Eq. \eqref{eq:Ineq6StoqCond}.
For the blue line $(\pi-\varphi = \theta)$, we observe that the segment which extends into the orange region grows with increasing $n$.}
\label{fig:I6Scatter}
\end{figure*}
\begin{figure*}[h]
\centering

\subfloat[$n=10$]{%
  \includegraphics[width=0.24\textwidth]{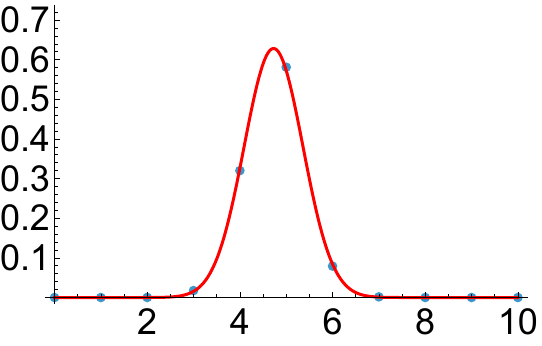}
}
\hfill
\subfloat[$n=20$]{%
  \includegraphics[width=0.24\textwidth]{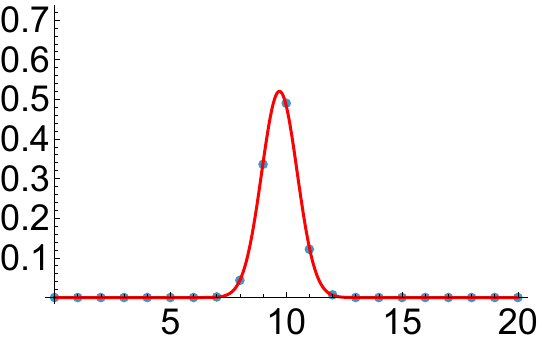}
}
\hfill
\subfloat[$n=30$]{%
  \includegraphics[width=0.24\textwidth]{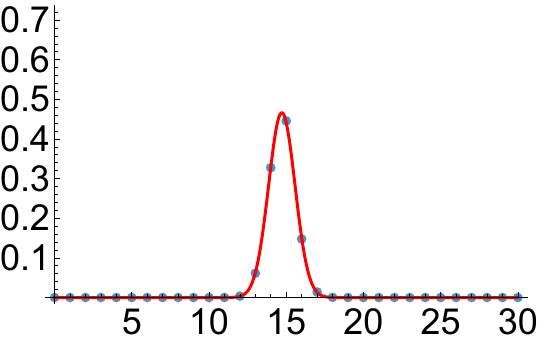}
}
\hfill
\subfloat[$n=50$]{%
  \includegraphics[width=0.24\textwidth]{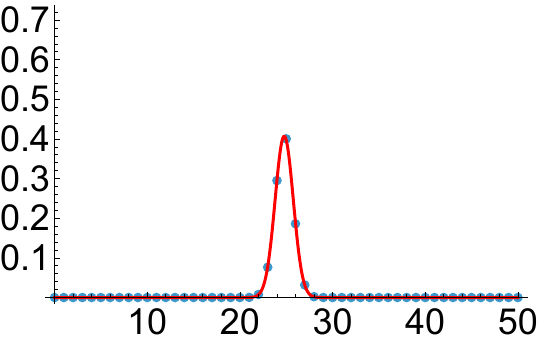}
}
\caption{Plots of coefficients of eigenvectors from numerically optimized Bell operators of table \ref{tb:OptGap}. The blue dots represent $\psi_i^2$, where $\psi_i$ is the $i-$th element of the ground state eigenvector up to normalization. 
The red line represent the fitted normal distribution.}
\label{fig:OptGS}
\end{figure*}
\vspace{-20cm}

\begin{table*}[h]
  \centering
  \begin{tabular}{|c|c|c|c|c|}
    \hline
    $n$ & Optimized $\vec{c}$ & $\vec{\alpha}$ & Opt. $\beta_Q/\beta_C$ & $\beta_Q/\beta_C$ Eq. \eqref{eq:Ineq6Op}\\ 
    \hline
    10 & $(6.95 \mathrm{e}{-2}, 7.11 \mathrm{e}{-2}, 1.24 \mathrm{e}{-5}, 1.71 \mathrm{e}{-3}, 5.50 \mathrm{e}{-1})$ &  (-2.20, 0.00171, 0.548, -1.10, 0.550) & 1.0541 & 1.05442 \\
    20 & $(4.44\mathrm{e}{-2},4.96\mathrm{e}{-2},0,2.52\mathrm{e}{-3},7.78\mathrm{e}{-1})$ &  (-3.09, 0.00252, 0.772, -1.55, 0.778) & 1.09688 & 1.09762\\
    30 & $(2.10\mathrm{e}{-2},2.03\mathrm{e}{-2},1.16\mathrm{e}{-5}, 4.22\mathrm{e}{-3},5.18\mathrm{e}{-1})$ & (-2.08, 0.00422, 0.519, -1.04, 0.518) &  1.11984 & 1.12005\\
    50 & $(2.00\mathrm{e}{-2},1.88\mathrm{e}{-2},0,-3.78\mathrm{e}{-1},8.20\mathrm{e}{-1})$ & (-2.90, -0.378, 0.821, -1.64, 0.820) &  1.14266 & 1.14482\\
    \hline
  \end{tabular}
  \caption{Values for numerically optimized stoquastic Bell operators at $(\varphi,\theta)=(\pi/6,5\pi/6)$.}
  \label{tb:OptGap}
\end{table*}

\clearpage
\pagestyle{plain}
\bibliography{ArXivReferences}

\clearpage
\onecolumngrid
\appendix

\section{Bell scenario}\label{app:BellScenario}

The typical Bell scenario consists of $n$ spacelike separated parties, each labelled by some index $i$. Each party has its own quantum system  
and can select an input $x_i$ from an set $\mathcal{X}$.
In response, the system produces an output $a_i$ selected from an set $\mathcal{A}$ with some probability.
The input and output sets are chosen to be the same for all parties. 
In this work, we focus on the scenario where both $\mathcal{X}$ and $\mathcal{A}$ have two elements, labeled $\{0,1\}$ and  $\{-1,+1\}$ respectively. 
The general probability distribution that characterizes all input and output correlations of all parties is denoted as $P(\vec{a}|\vec{x})$ with $\vec{a}=(a_0,\ldots,a_{n-1})$ and $\vec{x}=(x_0,\ldots,x_{n-1})$.
In this context, the probability distribution itself is usually referred to as a correlation.
In classical physics, there is the condition the output $a_i$ of party $i$ can only depend on the input $x_i$ of party $i$.
Additionally, all parties may have a shared source of randomness $p(\lambda)$. 
Classical probability distributions must thus be of the form 
\begin{equation}
\begin{aligned}
P_C(\vec{a}\mid \vec{x})
&= \int_{\Lambda} d\lambda \,p(\lambda)
   \prod_{i=0}^{n-1} P(a_i \mid x_i,\lambda). 
   \label{eq:LocalDistrApp}
\end{aligned}
\end{equation}
The parameter $\lambda$ is also called the hidden variable and this model for classical probability distributions is referred to as the local hidden variable model. 

In the quantum setting, probabilities of outcomes are given by the Born rule as follows
\begin{align}
    P_Q(\vec{a}|\vec{x}) = \mathrm{Tr}\!\left(
\rho \; \bigotimes_{j=0}^{n-1} \Pi_{a_j}(x_j)
\right),
\label{eq:QuantumProb}
\end{align}
where $\rho$ is the quantum state and $\Pi_{a_j}(x_j)$ are the local projection operators associated with measurement outcome $a_j$ at the local subsystem $j$.
If the quantum state $\rho$ was separable, the correlations would be classical. 
However, for some entangled states, correlations which cannot be written in the form of Eq. (\ref{eq:LocalDistrApp}) can be generated.

For a fixed input set $\mathcal{X}$ and a fixed output set $\mathcal{A}$, let us denote the set of classical correlations as $\mathcal{C}$ and the set of quantum correlations as $\mathcal{Q}$.
As the set of quantum correlations is strictly larger, we have the inclusion $\mathcal{C} \subsetneq \mathcal{Q}$~\cite{brunner_bell_2014}.
As $\mathcal{C}$ is a polytope, cf. Sec. \ref{sec:Polyhedron}, it can equivalently be characterized by a set of half spaces of the form 
\begin{equation}
   \beta_i \leq \sum_{\vec{a},\vec{x}} (\alpha_{\vec{a},\vec{x}})_i P(\vec{a}|\vec{x}),
   \label{eq:bellHplane}
\end{equation}
where $P$ is a conditional probability distribution over $\vec{a} = (a_0, ...,a_{n-1})$ given $\vec{x} = (x_0, ...,x_{n-1})$, $(\alpha_{\vec{a},\vec{x}})_i \in \mathbbm{R}$ are the Bell coefficients, $\beta_i$ is the classical bound and the index $i$ labels the various half space separations.
If a probability distribution achieves a value smaller than $\beta_i$ for some $i$, i.e. $\beta_i > \sum_{\vec{a},\vec{x}}(\alpha_{\vec{a},\vec{x}})_i P(\vec{a}|\vec{x})$, then the probability distribution $P(\vec{a}|\vec{x})$ cannot be of the form of Eq. \eqref{eq:LocalDistrApp}.
The inequalities of \eqref{eq:bellHplane}, also known as Bell inequalities, can thus be used to separate classical distributions from the quantum ones. 
The associated functional $\beta(P) \coloneqq \sum_{\vec{a},\vec{x}}(\alpha_{\vec{a},\vec{x}})_i P(\vec{a}|\vec{x})$ is called the Bell functional.

We define the expected value of a single random variable as 
\begin{equation}
\langle A_i(x_i) \rangle := \sum_{a_i \in \mathcal{A}} a_i \, P(a_i|x_i),
\end{equation}
and more generally, the $K$-body correlator as
\begin{equation}
\langle A_{i_1}(x_{i_1})\ldots A_{i_K}(x_{i_K})\rangle := \sum_{a_{i_1}, \ldots, a_{i_K} \in \mathcal{A}} a_{i_1}\cdots a_{i_K} \, P(a_{i_1}\ldots a_{i_K}|x_{i_1}\ldots x_{i_K}),
\end{equation}
where by $A_{i_j}(x_{i_j})$ we denote the random variable corresponding to output $a_{i_j}$ given input $x_{i_j}$.
Since we are working with random variables $a_1,\ldots,a_n$, which only take values $\{+1,-1\}$, the $K$-body correlators and the probability distributions $P(\vec{a}|\vec{x})$ are related through the following discrete Fourier transform
\begin{equation}
\label{eq:probs2expectations}
    \langle A_{i_1}(x_{i_1})\ldots A_{i_K}(x_{i_K})\rangle 
    =  \sum_{a_{i_1}, \ldots, a_{i_K}} (-1)^{\sum_{j=1}^{K}a_{i_j}}P(a_{i_1}\ldots a_{i_K}|x_{i_1}\ldots x_{i_K}).
    \nonumber
\end{equation}
As the discrete Fourier transform is invertible, this induces a one-to-one correspondence between the $K$-body correlators and the probability distributions $P(\vec{a}|\vec{x})$.
The number $K$, which denotes the number of random variables  $A_{i_j}(x_{i_j})$ appearing in the expected value, will be referred to as the order of the correlator. 
Instead of expressing the Bell inequality in terms of probabilities, we can thus express it as a linear combination of the correlators. 
In terms of correlators, the Bell functional is of the form 
\begin{equation}
\begin{aligned}
\beta(P) = \sum_{k=0}^{n} \sum_{\substack{i_1,\ldots,i_k}} \sum_{x_{i_1},\ldots,x_{i_k}} \beta_{i_1,\ldots,i_k}^{x_{i_1},\ldots,x_{i_k}} \langle A_{i_1}(x_{i_1}) \cdots A_{i_k}(x_{i_k}) \rangle,
\end{aligned}
\label{eq:BellFuncCorr}
\end{equation}
where $\beta_{i_1,\ldots,i_k}^{x_{i_1},\ldots,x_{i_k}} \in \mathbbm{R}$ are the Bell coefficients in correlator space.
If the correlations come from quantum probability distributions, they are of the form 
\begin{equation}
\langle A_{i_1}(x_{i_1}) \cdots A_{i_k}(x_{i_k}) \rangle = \mathrm{Tr}\!\left( \rho \; \mathbf{M}_{x_{i_1}}^{i_1} \otimes \cdots \otimes \mathbf{M}_{x_{i_k}}^{i_k} \right),
\label{eq:QuantumCorr}
\end{equation}
where $\mathbf{M}_{x_{i_j}}^{i_j}$ is the local measurement operator acting on subsystem $i_j$ with measurement setting $x_{i_j}$.
Eq. \eqref{eq:BellFuncCorr} can thus be written as $\mathrm{Tr}(\rho \mathbf{B})$, where the Bell operator $\mathbf{B}$ is defined as
\begin{equation}
\mathbf{B} \coloneqq \sum_{k=0}^{n} \sum_{\substack{i_1,\ldots,i_k }} \sum_{x_{i_1},\ldots,x_{i_k}} \beta_{i_1,\ldots,i_k}^{x_{i_1},\ldots,x_{i_k}} \; \mathbf{M}_{x_{i_1}}^{i_1} \otimes \cdots \otimes \mathbf{M}_{x_{i_k}}^{i_k}.
\label{eq:BellOpDef}
\end{equation}
In the case of PI Bell operators, we can write it as 
\begin{equation}
\mathbf{B} = \sum_{k=0}^{n} \sum_{x_{1},\ldots,x_{k}} \alpha_{x_1,\ldots,x_k} \; \mathbf{S}_{x_1,\ldots,x_k},
\label{eq:PIBellOpApp}
\end{equation}
where $\alpha_{x_1,\ldots,x_k} \in \mathbbm{R}$ are the Bell coefficients and $\mathbf{S}_{x_1,\ldots,x_k}$ denotes the $k$-body permutationally invariant measurement operator defined by
\begin{equation}
\mathbf{S}_{x_1,\ldots,x_k} \coloneqq \sum_{\substack{i_1,\ldots,i_k=0 \\ \text{all distinct}}}^{n-1} \mathbf{M}_{x_1}^{i_1} \otimes \cdots \otimes \mathbf{M}_{x_k}^{i_k}.
\label{eq:PIMeasOpApp}
\end{equation}
Since Bell operators are Hermitian operators, they can be interpreted as Hamiltonians and the value of the Bell functional can then be interpreted as the expected value of the corresponding Hamiltonian. 
Moreover, finding the state with the largest gap to the classical bound is thus equivalent to finding the ground state of the Hamiltonian. 

\section{The local polytope}

One way to characterize polytopes is through the vertex-representation, defined as follows:
for a set of points in a $d-$dimensional real Euclidean space $V = \{ \vec{v}_1, \vec{v}_2, \dots, \vec{v}_m \} \subset \mathbb{R}^d$, the polytope $\mathcal{P}$ spanned by $V$ is the convex hull of the set of points in $V$, i.e.
\begin{align}
\mathcal{P} & \coloneqq \left\{ \sum_{i=1}^m \lambda_i \vec{v}_i \ \middle|\ 
      \vec{v}_i \in V ,\ \lambda_i \ge 0,\ \sum_{i=1}^m \lambda_i = 1 \right\}. \nonumber
\end{align}
Alternatively, a polytope can be characterized by the set of bounding hyperplanes, known as the half-space representation, as follows:
\begin{align}
\mathcal{P} &= \bigcap_{j=1}^k \left\{ x \in \mathbb{R}^d \ \middle|\ a_j^\top x \le b_j \right\}. \nonumber
\end{align}
It has been proven that the vertex and hyperplane descriptions are equivalent, and therefore either representation can always be converted into the other~\cite{ziegler_lectures_1995}. 
Let $H_j$ denote the halfspace associated with $a_j^\top x \le b_j$.
The corresponding hyperplane $a_j^\top x = b_j$ is redundant if all points satisfying all other half-spaces $H_i$ automatically satisfy $H_j$.
A hyperplane is called irredundant if it is not redundant. 
The irredundant $(d-1)$-dimensional hyperplanes that bound the polytope are called facets. 
If the dimension of the irredundant hyperplane bounding the polytope is smaller than $d-1$, it is called a face. 

A prominent example of a polytope in the field of quantum information is the local polytope. 
Let us consider $n$ spacelike separated parties labelled by $i \in [n]$, each with a system which takes inputs $x_i \in \mathcal{X}$ and outputs $a_i \in \mathcal{A}$.
The local polytope is the set of all probability distributions of the form 
\begin{equation}
\begin{aligned}
P_C(\vec{a}\mid \vec{x})
&= \int_{\Lambda}d\lambda\; p(\lambda)
   \prod_{i=0}^{n-1} P(a_i \mid x_i,\lambda) ,
\end{aligned}
\label{eq:localprob}
\end{equation}
where $\vec{a}=(a_1,...,a_n)$ and $\vec{x} = (x_1,...,x_n)$ and $p(\lambda)$ being the probability distribution over the hidden variable $\lambda$.
For more details and how it relates to the quantum set, we refer the reader to App. \ref{app:BellScenario}.
We can view the probability distribution $P_C(\vec{a}\mid \vec{x})$ as a vector, where each entry of the vector is the probability for a specific $\vec{a} \in \mathbbm{R}^{|\mathcal{A}|}$ to occur given a specific $\vec{x} \in \mathbbm{R}^{|\mathcal{X}|}$. 
For example, for the bipartite binary-input binary-output scenario, the vector is given by
\begin{equation}
    \begin{aligned}
    \vec{P}= \bigg(&P(00|00),P(01|00),P(10|00),P(11|00),\\
    &P(00|01),P(01|01),...,P(10|11),P(11|11) \bigg).
    \end{aligned}
    \nonumber
\end{equation}
In this view, is clear that the set of correlations of the form Eq. \eqref{eq:localprob} forms a polytope for a fixed $n,\mathcal{A}, \mathcal{X}$.
The vertices of this polytope are the local deterministic probability distributions~\cite{brunner_bell_2014}, i.e. distributions of the form $P(a_i|x_i) = \delta(a_i = f_i(x_i))$ with $f_i: \mathcal{X} \rightarrow \mathcal{A}$.
The facets of the polytope are the hyperplanes that separate all local probability distributions from the nonlocal probability distributions, i.e. probability distribution that do not satisfy Eq. \eqref{eq:localprob}. 
These facets are also known as tight Bell inequalities. 

\section{Matrix elements of PI measurement operators}\label{app:HighToSingle}

Let us start by reminding ourselves that single-body measurement operators are defined as
\begin{equation}
     \mathbf{S}_{x_1} := \sum_{i=0}^{n-1} \mathbf{M}_{x_1}^{i},
     \nonumber
\end{equation}
two-body measurement operators are defined as 
\begin{equation}
     \mathbf{S}_{x_1 x_2}:=\sum_{\substack{i_1,i_2=0\\i_1\neq i_2}}^{n-1} \mathbf{M}_{x_1}^{i_1} \otimes \mathbf{M}_{x_2}^{i_2},
     \nonumber
\end{equation}
and three-body measurement operators are defined as
\begin{equation}
     \mathbf{S}_{x_1 x_2 x_3}:=\sum_{\substack{i_1,i_2,i_3=0\\ \text{all distinct} }}^{n-1} \mathbf{M}_{x_1}^{i_1} \otimes \mathbf{M}_{x_2}^{i_2} \otimes \mathbf{M}_{x_3}^{i_3},
     \nonumber
\end{equation}
where
\begin{equation}
\mathbf{M}_0=\cos(\varphi) Z^{i}+ \sin(\varphi) X^{i}, \qquad \mathbf{M}_1=\cos(\theta) Z^{i}+ \sin(\theta) X^{i}.
\nonumber
\end{equation}
Here $X,Y,Z$ denote the Pauli X, Pauli Y and Pauli Z matrices. 
In general, any $K$-body PI Bell operator can be written as a linear combination of $K$-body PI Pauli operators, defined as 
\begin{align}
   \mathbf{S}_{P_1 P_2 \cdots P_K} \coloneq \sum_{\substack{i_1,...,i_K=0\\ \text{all distinct} }}^{n-1} P_1^{i_1} \otimes P_2^{i_2} \otimes \cdots \otimes P_K^{i_K},
   \label{eq:Paulicorrelator}
\end{align}
where $P_j \in \{X,Y,Z\}$.
To see the above, let us start by noting that the single-body permutationally invariant Bell operators can be written as 
\begin{align}
\mathbf{S}_0 &= \cos\varphi \, \mathbf{S}_Z + \sin\varphi \, \mathbf{S}_X, \nonumber \\
\mathbf{S}_1 &= \cos\theta \, \mathbf{S}_Z + \sin\theta \, \mathbf{S}_X, \nonumber
\end{align}
where
\begin{align}
\mathbf{S}_Z &= \sum_{i=0}^{n-1} Z^i, \nonumber \\
\mathbf{S}_X &= \sum_{i=0}^{n-1} X^i. \nonumber
\end{align}
The two-body permutationally invariant Bell operators can be written as 
\begin{align}
\mathbf{S}_{00} &= \cos^2\varphi \, \mathbf{S}_{ZZ} + 2 \cos\varphi \sin\varphi \, \mathbf{S}_{ZX} + \sin^2\varphi \, \mathbf{S}_{XX}, \nonumber\\
\mathbf{S}_{01} &= \cos\varphi \cos\theta \, \mathbf{S}_{ZZ} + (\cos\varphi \sin\theta + \sin\varphi \cos\theta) \mathbf{S}_{ZX} + \sin\varphi \sin\theta \, \mathbf{S}_{XX}, \nonumber\\
\mathbf{S}_{11} &= \cos^2\theta \, \mathbf{S}_{ZZ} + 2 \cos\theta \sin\theta \, \mathbf{S}_{ZX} + \sin^2\theta \, \mathbf{S}_{XX}. \nonumber
\end{align}
and the three-body permutationally invariant Bell operators can be written as 
\begin{align}
\mathbf{S}_{000} &= \cos^3\varphi \, \mathbf{S}_{ZZZ} + 3 \cos^2\varphi \sin\varphi \, \mathbf{S}_{ZZX} + 3 \cos\varphi \sin^2\varphi \, \mathbf{S}_{ZXX} + \sin^3\varphi \, \mathbf{S}_{XXX}, 
\nonumber\\
\mathbf{S}_{001} &= \cos^2\varphi \cos\theta \, \mathbf{S}_{ZZZ} + \left( 2 \cos\theta \cos\varphi \sin\varphi + \cos^2\varphi \sin\theta \right) \mathbf{S}_{ZZX} \nonumber\\
&\quad+ \left( 2 \sin\theta \sin\varphi \cos\varphi + \sin^2\varphi \cos\theta \right) \mathbf{S}_{ZXX} + \sin^2\varphi \sin\theta \, \mathbf{S}_{XXX},
\nonumber\\
\mathbf{S}_{011} &= \cos^2\theta \cos\varphi \, \mathbf{S}_{ZZZ} + \left( 2 \cos\varphi \cos\theta \sin\theta + \cos^2\theta \sin\varphi \right) \mathbf{S}_{ZZX} \nonumber\\
&\quad+ \left( 2 \sin\varphi \sin\theta \cos\theta + \sin^2\theta \cos\varphi \right) \mathbf{S}_{ZXX} + \sin^2\theta \sin\varphi \, \mathbf{S}_{XXX},
\nonumber\\
\mathbf{S}_{111} &= \cos^3\theta \, \mathbf{S}_{ZZZ} + 3 \cos^2\theta \sin\theta \, \mathbf{S}_{ZZX} + 3 \cos\theta \sin^2\theta \, \mathbf{S}_{ZXX} + \sin^3\theta \, \mathbf{S}_{XXX}. \nonumber
\end{align}
We note that for the PI Pauli operators of Eq. \eqref{eq:Paulicorrelator}, only the number of distinct Pauli operators $X,Y,Z$ is relevant and not the order.
For example $\mathbf{S}_{ZXX} = \mathbf{S}_{XZX} = \mathbf{S}_{XXZ}$.
The two-body PI Pauli operators can be written in terms of the single-body PI Pauli operators using the the following relations:
\begin{align}
    \mathbf{S}_{ZZ} &= (\mathbf{S}_Z)^2 -n \mathbbm{I}, \nonumber\\
    \mathbf{S}_{ZX} &= \mathbf{S}_Z \mathbf{S}_X - i \mathbf{S}_Y, \nonumber\\
    \mathbf{S}_{XX} &= (\mathbf{S}_X)^2 -n \mathbbm{I}. \nonumber
\end{align}
Similarly, the three-body PI Pauli operators can be written in terms of the single-body PI Pauli operators using the following relations:
\begin{align}
    \mathbf{S}_{ZZZ} &= (\mathbf{S}_Z)^3 + (2-3n)\mathbf{S}_Z ,\nonumber\\
    \mathbf{S}_{ZZX} &= (\mathbf{S}_Z)^2 \mathbf{S}_X - n \mathbf{S}_X- i \mathbf{S}_{Y} \mathbf{S}_Z - i \mathbf{S}_{Z} \mathbf{S}_{Y} ,\nonumber\\
    \mathbf{S}_{ZXX} &= \mathbf{S}_Z (\mathbf{S}_X)^2 +(2-n) \mathbf{S}_Z - 2i \mathbf{S}_Y \mathbf{S}_X ,\nonumber\\
    \mathbf{S}_{XXX} &= (\mathbf{S}_X)^3 + (2-3n)\mathbf{S}_X .\nonumber
\end{align}
The proof for the above relations can be found in App. \ref{app:ThreeToOnePIPauliProof}.
The matrix elements of the single-body PI Pauli operators are given by~\cite{tura_nonlocality_2015}.
\begin{align}
(\mathbf{S}_X)_{k,l} &= \Gamma(l)\delta_{k,l+1} + \Gamma(k)\delta_{l,k+1}, \nonumber \\ 
(\mathbf{S}_Y)_{k,l} &= i\Gamma(l)\delta_{k,l+1} - i \Gamma(k)\delta_{l,k+1}, \nonumber\\
(\mathbf{S}_Z)_{k,l} &= \Xi(k) \delta_{k,l}, \nonumber
\end{align}
where $\Gamma(x) = \sqrt{(x+1)(2J-x)}$ and $\Xi(x) = (2J-2x)$.
The matrix elements of products of two-body PI Pauli operators are
\begin{align}
    (\mathbf{S}_X \mathbf{S}_X)_{k,l} &=  \Gamma(l+1) \Gamma(l) \delta_{k,l+2} + (\Gamma(l)^2 + \Gamma(l-1)^2) \delta_{k,l} + \Gamma(k+1) \Gamma(k) \delta_{l,k+2},
    \nonumber\\
    (\mathbf{S}_Y \mathbf{S}_Y)_{k,l} &=  - \Gamma(l+1) \Gamma(l) \delta_{k,l+2} + (\Gamma(l)^2 + \Gamma(l-1)^2) \delta_{k,l} - \Gamma(k+1) \Gamma(k) \delta_{l,k+2},
    \nonumber\\
    (\mathbf{S}_X \mathbf{S}_Y)_{k,l} &= i \Gamma(l+1) \Gamma(l) \delta_{k,l+2} +i (\Gamma(l)^2 - \Gamma(l-1)^2) \delta_{k,l} - i \Gamma(k+1) \Gamma(k) \delta_{l,k+2},
    \nonumber\\
    (\mathbf{S}_Y \mathbf{S}_X)_{k,l} &= i \Gamma(l+1) \Gamma(l) \delta_{k,l+2} + i (\Gamma(l-1)^2 - \Gamma(l)^2) \delta_{k,l} - i \Gamma(k+1) \Gamma(k) \delta_{l,k+2}, \nonumber\\
   (\mathbf{S}_X \mathbf{S}_Z)_{k,l} &= \Xi(l)\Gamma(l) \delta_{k,l+1} + \Xi(l) \Gamma(k) \delta_{l,k+1},  \nonumber\\
   (\mathbf{S}_Z \mathbf{S}_X)_{k,l} &= \Xi(k)\Gamma(l) \delta_{k,l+1} + \Xi(k) \Gamma(k) \delta_{l,k+1}, \nonumber\\
   (\mathbf{S}_Y \mathbf{S}_Z)_{k,l} &= i \Xi(l)\Gamma(l) \delta_{k,l+1} -i \Xi(l) \Gamma(k) \delta_{l,k+1}, \nonumber\\
   (\mathbf{S}_Z \mathbf{S}_Y)_{k,l} &= i \Xi(k)\Gamma(l) \delta_{k,l+1} -i \Xi(k) \Gamma(k) \delta_{l,k+1}, \nonumber\\
   (\mathbf{S}_Z \mathbf{S}_Z)_{k,l} &= \Xi(k)^2 \delta_{k,l}.\nonumber
\end{align}
The matrix elements of products of three-body PI Pauli operators are
\begin{align}
    ((\mathbf{S}_Z)^3)_{k,l} =& \Xi(k)^3 \delta_{k,l}, \nonumber\\
    ((\mathbf{S}_Z)^2 \mathbf{S}_X)_{k,l} =& \Xi(k)^2 \Gamma(l) \delta_{k,l+1} + \Xi(k)^2 \Gamma(k) \delta_{l,k+1}, \nonumber\\
    (\mathbf{S}_Z (\mathbf{S}_X)^2)_{k,l} =& \Xi(k) \Gamma(l+1) \Gamma(l) \delta_{k,l+2} + \Xi(k) (\Gamma(l)^2 + \Gamma(l-1)^2) \delta_{k,l} + \Xi(k) \Gamma(k+1) \Gamma(k) \delta_{l,k+2}, \nonumber\\
    ((\mathbf{S}_X)^3)_{k,l} =& \Gamma(l+2) \Gamma(l+1) \Gamma(l) \delta_{k,l+3} + \left( \Gamma(l+1)^2 \Gamma(l) + \Gamma(l)^3 + \Gamma(l) \Gamma(l-1)^2 \right) \delta_{k,l+1} \nonumber\\
    & + \left( \Gamma(k+1)^2 \Gamma(k) + \Gamma(k)^3 + \Gamma(k) \Gamma(k-1)^2 \right) \delta_{l,k+1} + \Gamma(k+2) \Gamma(k+1) \Gamma(k) \delta_{l,k+3}, \nonumber\\
    (\mathbf{S}_X \mathbf{S}_Z^2)_{k,l} =& \Xi(l)^2\Gamma(l) \delta_{k,l+1} + \Xi(l)^2 \Gamma(k) \delta_{l,k+1},  \nonumber\\
    (\mathbf{S}_Y \mathbf{S}_Z^2)_{k,l} =& i \Xi(l)^2\Gamma(l) \delta_{k,l+1} -i \Xi(l)^2 \Gamma(k) \delta_{l,k+1}, \nonumber\\
     (\mathbf{S}_Y \mathbf{S}_Z \mathbf{S}_X)_{k,l} =& i \Xi(l+1) \Gamma(l+1) \Gamma(l) \delta_{k,l+2} +i\left(\Xi(l-1) \Gamma(l-1)^2-\Xi(l+1) \Gamma(l)^2\right)\delta_{k,l} \nonumber\\
     &-i \Xi(k+1) \Gamma(k+1) \Gamma(k) \delta_{l,k+2}, \nonumber\\
     (\mathbf{S}_Z \mathbf{S}_Y \mathbf{S}_X)_{k,l} =& i \Xi(k) \Gamma(l+1) \Gamma(l) \delta_{k,l+2} + i \Xi(k) (\Gamma(l-1)^2 - \Gamma(l)^2) \delta_{k,l} \nonumber\\
     &- i \Xi(k) \Gamma(k+1) \Gamma(k) \delta_{l,k+2}, \nonumber\\
      (\mathbf{S}_Y \mathbf{S}_X \mathbf{S}_X)_{k,l} =& i \Gamma(l+2) \Gamma(l+1) \Gamma(l) \delta_{k,l+3} + i (\Gamma(l)^3 - \Gamma(l+1)^2 \Gamma(l) + \Gamma(l)\Gamma(l-1)^2) \delta_{k,l+1}, \nonumber\\
      & - i (\Gamma(k)^3 + \Gamma(k+1)^2 \Gamma(k) - \Gamma(k)\Gamma(k-1)^2) \delta_{l,k+1} - i \Gamma(k+2) \Gamma(k+1) \Gamma(k) \delta_{l,k+3}.\nonumber
\end{align}
Using this, we can now calculate the nonzero matrix elements of the PI measurement operators. 
Since the elements on the of diagonals are real and since Bell operators must be Hermitian, the Permutationally invariant Bell operators are symmetric operators. 
Therefore, we only need to focus on the diagonal elements $(k,l=k+d)$, where $d$ is some positive integer.
For the two-body measurement operators, the elements for the first off-diagonals are 
\begin{align}
\mathbf{S}_0^{k,k+1} &= \Gamma(k) \sin (\phi),\nonumber\\ 
\mathbf{S}_1^{k,k+1} &= \Gamma(k) \sin (\theta), \nonumber\\ 
\mathbf{S}_{00}^{k,k+1} &=\Gamma(k) (-2 k+n-1) \sin (2 \phi ), \nonumber\\ 
\mathbf{S}_{01}^{k,k+1} &= \Gamma(k) (-2 k+n-1) \sin ( \phi + \theta ), \nonumber\\ 
\mathbf{S}_{11}^{k,k+1} &= \Gamma(k) (-2 k+n-1) \sin (2 \theta ), \nonumber
\end{align}
and the elements of the second off-diagonals are
\begin{align}
\mathbf{S}_{00}^{k,k+2} &= \Gamma(k) \Gamma(k+1)  \sin ^2(\phi ), \nonumber\\ 
\mathbf{S}_{01}^{k,k+2} &=  \Gamma(k)  \Gamma(k+1) \sin (\theta )  \sin (\phi ), \nonumber\\ 
\mathbf{S}_{11}^{k,k+2} &= \Gamma(k) \Gamma(k+1) \sin ^2(\theta ). \nonumber
\end{align}
All other off-diagonals are $0$.
For the three-body measurement operators, the elements of the first off-diagonals are
\begin{align}
\mathbf{S}_{000}^{k,k+1} = &3 \Gamma(k) \sin (\phi ) \left(\left(4 k^2-4 k (n-1)+n^2-3 n+2\right) \cos ^2(\phi )+k (-k+n-1) \sin ^2(\phi )\right),\nonumber\\ 
\mathbf{S}_{001}^{k,k+1} =& \Gamma(k) \bigg[\left(4 k^2-4 k (n-1)+n^2-3 n+2\right) \cos (\phi ) \left(\sin (\theta ) \cos (\phi )+2 \cos (\theta ) \sin (\phi ) \right) \nonumber \\
& - 3 k \sin (\theta ) (k-n+1) \sin ^2(\phi )\bigg],\nonumber\\ 
\mathbf{S}_{011}^{k,k+1} =& \Gamma(k) \bigg[\cos (\theta ) \left(4 k^2-4 k (n-1)+n^2-3 n+2\right) \left(2 \sin (\theta ) \cos (\phi )+\cos (\theta ) \sin (\phi )\right) \nonumber \\
&-3 k \sin ^2(\theta ) (k-n+1) \sin (\phi )\bigg],\nonumber\\ 
\mathbf{S}_{111}^{k,k+1} =& 3 \sin (\theta ) \Gamma(k) \left(\cos ^2(\theta ) \left(4 k^2-4 k (n-1)+n^2-3 n+2\right)+k \sin ^2(\theta ) (-k+n-1)\right),\nonumber
\end{align}
the elements of the second off-diagonals are
\begin{align}
\mathbf{S}_{000}^{k,k+2} &= 3 \Gamma(k+1)  \Gamma(k)  (-2 k+n-2) \sin ^2(\phi ) \cos (\phi ),\nonumber\\ 
\mathbf{S}_{001}^{k,k+2} &= \Gamma(k+1)  \Gamma(k) (-2 k+n-2) \sin (\phi ) (2 \sin (\theta ) \cos (\phi )+\cos (\theta ) \sin (\phi )),\nonumber\\ 
\mathbf{S}_{011}^{k,k+2} &= \sin (\theta ) \Gamma(k+1)  \Gamma(k) (-2 k+n-2) (\sin (\theta ) \cos (\phi )+2 \cos (\theta ) \sin (\phi )),\nonumber\\ 
\mathbf{S}_{111}^{k,k+2} &= 3 \sin ^2(\theta ) \cos (\theta ) \Gamma(k+1)  \Gamma(k) (-2 k+n-2),\nonumber
\end{align}
and the elements of the third off-diagonals are
\begin{align}
\mathbf{S}_{000}^{k,k+3} &= \Gamma(k+2) \Gamma(k+1)  \Gamma(k) \sin ^3(\phi ),\nonumber\\ 
\mathbf{S}_{001}^{k,k+3} &= \sin (\theta ) \Gamma(k+2) \Gamma(k+1)  \Gamma(k) \sin ^2(\phi ),\nonumber\\ 
\mathbf{S}_{011}^{k,k+3} &= \sin ^2(\theta ) \Gamma(k+2) \Gamma(k+1)  \Gamma(k) \sin (\phi ),\nonumber\\ 
\mathbf{S}_{111}^{k,k+3} &= \sin ^3(\theta ) \Gamma(k+2) \Gamma(k+1)  \Gamma(k).\nonumber
\end{align}

\section{Higher to single-body PI Pauli operators identities}

To calculate the matrix elements of the higher order PI Pauli operators, it is convenient to write them in terms of product of single body PI Pauli operators as demonstrated in App. \ref{app:HighToSingle}.
The decompositions of two-body and three-body PI Pauli operators in terms of single-body Pauli operators is given in the following subsections. 

\subsection{Two-body to single-body PI Pauli operator proof \label{app:TwoToOnePIPauliProof}}
\begin{align}
    \mathbf{S}_{ZZ} = \sum_{i \neq j} Z^i Z^j =  \sum_{i, j} Z^i Z^j - \sum_{i} Z^i Z^i = \sum_{i, j} Z^i Z^j - \sum_i \mathbbm{I} = \mathbf{S}_Z \mathbf{S}_Z - n \mathbbm{I} \nonumber\\
    \mathbf{S}_{ZX} = \sum_{i \neq j} Z^i X^j =  \sum_{i, j} Z^i X^j - \sum_{i} Z^i X^i = \sum_{i, j} Z^i X^j - i \sum_i Y = \mathbf{S}_Z \mathbf{S}_X - \mathbf{S}_Y \nonumber\\
    \mathbf{S}_{XX} = \sum_{i \neq j} X^i X^j =  \sum_{i, j} X^i X^j - \sum_{i} X^i X^i = \sum_{i, j} X^i X^j - \sum_i \mathbbm{I} = \mathbf{S}_X \mathbf{S}_X - n \mathbbm{I} \nonumber
\end{align}

\subsection{Three-body to single-body PI Pauli operator proof \label{app:ThreeToOnePIPauliProof}}

Let us first define 

\begin{align}
    \mathbf{S}'_{P_1 P_2 \cdots P_k} \coloneq \sum_{i} P_1^{i}  P_2^{i}  \cdots  P_k^{i}, \nonumber
\end{align}
where for each term in the sum, every operator in the Pauli string acts on the same qubit. 
The proof for how to rewrite the three-body Pauli-operators in terms of only single-body Pauli-operators is as follows

\begin{align}
    \mathbf{S}_{P_1 P_2 P_3} &= \sum_{|\{i,j,k\}|=3} P_1^{i}  P_2^{j}  P_3^{k} \nonumber\\
    &= \sum_{|\{i,j,k\}|=3} P_1^{i}  P_2^{j}  P_3^{k} + \sum_{i \neq j, k=i,j} P_1^{i}  P_2^{j}  P_3^{k}
    - \sum_{i \neq j, k=i,j} P_1^{i} P_2^{j}  P_3^{k} \nonumber
\end{align}
Define
\begin{align}
    \omega &= \sum_{i \neq j, k=i,j} P_1^{i}  P_2^{j}  P_3^{k} = \sum_{i \neq j} P_1^{i}  P_2^{j}  P_3^{i} + P_1^{i}  P_2^{j}  P_3^{j} \nonumber  
\end{align}
Let the first term in $\omega$ be $\omega_1$ and the second term be $\omega_2$.
\begin{align}
    \omega_1 &= \sum_{i \neq j} P_1^{i}  P_2^{j}  P_3^{i} + \sum_{i} P_1^{i}  P_3^{i}  P_2^{i} - \sum_{i} P_1^{i}  P_3^{i}  P_2^{i} \nonumber\\
    &= \sum_{i} P_1^{i}  P_3^{i}  \sum_{j} P_2^{j} - \sum_{i} P_1^{i}  P_3^{i} P_2^{i} \nonumber\\
    &= \mathbf{S}'_{P_1 P_3} S_{P_2} - \mathbf{S}'_{P_1 P_3 P_2} \nonumber
\end{align}

\begin{align}
    \omega_2 &= \sum_{i \neq j} P_1^{i}  P_2^{j}  P_3^{j} + \sum_{i} P_1^{i}  P_2^{i}  P_3^{i} - \sum_{i} P_1^{i}  P_2^{i}  P_3^{i} \nonumber\\
    &= \sum_{i} P_1^{i}  \sum_{j} P_2^{j}  P_3^{j} - \sum_{i} P_1^{i}  P_2^{i} P_3^{i} \nonumber\\
    &= \mathbf{S}_{P_1} \mathbf{S}'_{P_2 P_3} - \mathbf{S}'_{P_1 P_2 P_3} \nonumber
\end{align}

\begin{align}
    \mathbf{S}_{P_1 P_2 P_3} + \omega &=  \sum_{i \neq j, k} P_1^{i}  P_2^{j}  P_3^{k} \nonumber\\
    &=  \sum_{i \neq j, k} P_1^{i}  P_2^{j}  P_3^{k} +  \sum_{i} P_1^{i}  P_2^{i} \sum_k  P_3^{k} - \sum_{i} P_1^{i}  P_2^{i} \sum_k  P_3^{k}\nonumber \\ 
    &= \sum_{i, j, k} P_1^{i}  P_2^{j}  P_3^{k} - \sum_{i} P_1^{i}  P_2^{i} \sum_k  P_3^{k} \nonumber\\
    &= \mathbf{S}_{P_1} \mathbf{S}_{P_2} \mathbf{S}_{P_3} - \mathbf{S}'_{P_1 P_2} \mathbf{S}_{P_3} \nonumber
\end{align}

Combining all of the above, we thus get that the three-body Pauli-operators thus takes the following form in terms of single-body Pauli-operators.

\begin{align}
    \mathbf{S}_{P_1 P_2 P_3} = \mathbf{S}_{P_1} \mathbf{S}_{P_2} \mathbf{S}_{P_3} - \mathbf{S}'_{P_1 P_2} \mathbf{S}_{P_3} - \mathbf{S}'_{P_1 P_3} \mathbf{S}_{P_2} - \mathbf{S}_{P_1} \mathbf{S}'_{P_2 P_3}  + \mathbf{S}'_{P_1 P_2 P_3} + \mathbf{S}'_{P_1 P_3 P_2} \nonumber
\end{align}

\section{Two-body Hyperplanes}\label{app:TwoBodyHPlanes}

We remind ourselves that 
$\mathbf{B} = \alpha_0 \mathbf{S}_{0} + \alpha_1 \mathbf{S}_{1} + \alpha_2 \mathbf{S}_{00} + \alpha_3 \mathbf{S}_{01} + \alpha_4 \mathbf{S}_{11}$.
The stoquasticity conditions on the first off-diagonals $\mathbf{B}_{k,k+1} \leq 0$ after dividing out the common positive prefactor $\Gamma(k)$ is given by
\[
\alpha_0 \sin (\phi) + \alpha_1 \sin (\theta) + \alpha_2 (-2 k+n-1) \sin (2 \phi ) + \alpha_3 (-2 k+n-1) \sin ( \phi + \theta ) + \alpha_4 (-2 k+n-1) \sin (2 \theta ) \leq 0.\nonumber
\] 
The last three terms depend on $k$, with $k \in \{0,\ldots,n-1\}$. 
The hyperplane at $k=0$ is given by
\begin{equation}
\label{eq:MinKHyperplane}
\alpha_0 \sin (\phi) + \alpha_1 \sin (\theta) + \alpha_2 (n-1) \sin (2 \phi ) + \alpha_3 (n-1) \sin ( \phi + \theta ) + \alpha_4 (n-1) \sin (2 \theta ) \leq 0,
\end{equation}
and the hyperplane at $k=n-1$ is given by
\begin{equation}
\label{eq:MaxKHyperplane}
\alpha_0 \sin (\phi) + \alpha_1 \sin (\theta) + \alpha_2 (1-n) \sin (2 \phi ) + \alpha_3 (1-n) \sin ( \phi + \theta ) + \alpha_4 (1-n) \sin (2 \theta ) \leq 0.
\end{equation}
We notice that all other hyperplanes can be written as conical combinations of the hyperplanes of Eq. \eqref{eq:MinKHyperplane} and Eq. \eqref{eq:MaxKHyperplane}.
Therefore, the hyperplanes at $k=0$ and $k=n-1$ are the only irredundant ones from the family of first off-diagonal conditions.
This gives the first two inequalities of Ineq. \eqref{eq:TwoBodyHyperplanes} in Section \ref{sec:2BodyRes}.
The inequalities corresponding to the second off-diagonal is of the form 
\[
 \alpha_2 \sin ^2(\phi )+ \alpha_3 \sin ( \phi) \sin( \theta ) + \alpha_4 \sin ^2(\theta ) \leq 0,\nonumber
\] 
after dividing out the common prefactors $\Gamma(k) \Gamma(k+1)$.
This results in only one stoquastic condition of the form 
\[
\alpha_2 s^2_2+ \alpha_3 s^2_3+ \alpha_4 s^2_4 \leq 0, \nonumber
\]
giving the last inequality of Ineq. \eqref{eq:TwoBodyHyperplanes} in Section \ref{sec:2BodyRes}. 

\section{Two-body Polyhedron}\label{app:TwoBodyPolyhedron}

As shown in App. \ref{app:TwoBodyHPlanes}, the polyhedron of the two-body Bell operator $\mathbbm{S}^{n,2}_{\varphi,\theta}$ is characterized by the following inequalities
\begin{equation}
\begin{aligned}
H_1 :& \quad \alpha_0 \sin \phi + \alpha_1 \sin \theta + \alpha_2 (n-1) \sin (2\phi)
+ \alpha_3 (n-1) \sin (\phi+\theta) + \alpha_4 (n-1) \sin (2\theta) \le 0,\\
H_2 :& \quad \alpha_0 \sin \phi + \alpha_1 \sin \theta + \alpha_2 (1-n) \sin (2\phi)
+ \alpha_3 (1-n) \sin (\phi+\theta) + \alpha_4 (1-n) \sin (2\theta) \le 0,\\
H_3 :& \quad \alpha_2 \sin^2 \phi + \alpha_3 \sin \phi \sin \theta + \alpha_4 \sin^2 \theta \le 0.
\label{eq:TwoBodyPolEq}
\end{aligned}
\end{equation}
To see that these three inequalities are irredundant, let us consider the following points
\begin{equation}
\begin{aligned}
    \vec{p}_1 &= (0,0,1,0,\csc (\theta ) \sec (\theta ) \sin (\phi ) (-\cos (\phi ))), \nonumber \\
    \vec{p}_2 &= (-2 (n-1) \csc (\theta ) \sin (\theta -\phi ),0,1,0,-\csc ^2(\theta ) \sin ^2(\phi )), \nonumber \\
    \vec{p}_3 &= (2 (n-1) \csc (\theta ) \sin (\theta -\phi ),0,1,0,-\csc ^2(\theta ) \sin ^2(\phi ) ), \nonumber
\end{aligned}
\end{equation}
and let us define the matrix 
\[
\mathbf{H} =\begin{pmatrix}
\sin\phi & \sin\theta & (n-1)\sin(2\phi) & (n-1)\sin(\phi+\theta) & (n-1)\sin(2\theta)  \\
\sin\phi & \sin\theta & (1-n)\sin(2\phi) & (1-n)\sin(\phi+\theta) & (1-n)\sin(2\theta) \\
0 & 0 & \sin^2\phi & \sin\phi\,\sin\theta & \sin^2\theta 
\end{pmatrix},
\]
such that $ \mathbf{H} \vec{\alpha} \leq 0$ gives us the system of equations of Eq. \eqref{eq:TwoBodyPolEq}, where $\vec{\alpha} = (\alpha_0,\alpha_1,\alpha_2,\alpha_3,\alpha_4)$.
We then have the following:
\begin{equation}
\begin{aligned}
    \mathbf{H} \vec{p}_1 &= (0,0,\sin (\phi ) (\sin (\phi )-\tan (\theta ) \cos (\phi ))^T, \nonumber \\
    \mathbf{H} \vec{p}_2 &= (0,-4 (n-1) \csc (\theta ) \sin (\phi ) \sin (\theta -\phi ),0)^T, \nonumber \\
    \mathbf{H} \vec{p}_3 &= (4 (n-1) \csc (\theta ) \sin (\phi ) \sin (\theta -\phi ),0,0)^T. \nonumber
\end{aligned}
\end{equation}
Thus for each point, there is only one equality which is non zero. 
The inequality can then be chosen to be violated by choosing either $+\vec{p}_i$ or $-\vec{p}_i$ where $i \in \{1,2,3\}$, such that the nonzero term is positive. 

The lines are the basis vectors of the kernel of $\mathbf{H}$.
Solving $\mathbf{H} \vec{\alpha} = 0$, we obtain the following expressions for the lines:
\begin{equation}
\begin{aligned}
    \vec{l}_1 &= (-\csc(\phi)\,\sin(\theta) , 1 , 0 , 0 , 0), \nonumber\\
    \vec{l}_2 &= (0,0,\sin ^2(\theta ) \csc ^2(\phi ),-2 \sin (\theta ) \csc (\phi ),1). \nonumber
\end{aligned}
\end{equation}
The dimension $d$ of the pointed cone, i.e. after factoring out the lineality space spanned by lines, is $\text{Rank}(\mathbf{H})=3$.
In a $d-$dimensional cone, every extreme ray arises as the intersection of exactly $d-1$ linearly independent hyperplanes. 
Therefore, the total number of extreme rays is bounded from above by $\binom{m}{d-1}$, where $m$ denotes the number of irredundant hyperplanes.
This bound follows because each choice of $d-1$ hyperplanes can intersect in at most one $1-$dimensional face, and not every such intersection necessarily lies in the cone.
Hence $\binom{m}{d-1}$ is a strict upperbound in general. 
For our case, as $m=3$ and $d=3$, the upperbound on the number of rays is $\binom{3}{2}=3$.
To find explicit expressions of the rays, we solve 
the system of equations
\begin{equation}
\label{eq:TwoBodyRaysSystems}
\begin{pmatrix}
\vec{h}_i^T\\
\vec{h}_j^T\\
\vec{l}_1^T \\
\vec{l}_2^T 
\end{pmatrix}
\vec{\alpha} = \vec{0},
\end{equation}
where $h_i^T$ and $h_j^T$ with $i,j \in \{1,2,3\}$ and $i \neq j$ are the $i$-th and $j$-th row of the matrix $\mathbf{H}$.
The conditions $\vec{l}_u \cdot \vec{\alpha} = 0$, with $u \in \{1,2\}$, are added as we are working directly in $\mathbbm{R}^5$, instead of projecting down to the pointed cone first. 
Since adding linear combination of lines to any ray $r_i$, i.e. $r_i +c_1 l_1 + c_2 l_2$ with $c_1,c_2 \in \mathbbm{R}$, is still a valid ray, we require $c_1,c_2=0$ to remove this extra degree of freedom. 
We note that strictly speaking, instead of requiring $l_i \vec{\alpha} = 0$, we could have required $l_i \vec{\alpha} = c_i$ with $c_i \in \mathbbm{R}$. 
This would not have affected our analysis however. 
Solving Eq. \eqref{eq:TwoBodyRaysSystems}, we obtain the following rays
\[
r_1 =
\left(
\begin{array}{ccccc}
 -\frac{(n-1) u_1 \csc (\theta ) (24 \sin (\theta -\phi )+4 \sin (3 (\theta -\phi ))-12 \sin (3 \theta -\phi )+\sin (5 \theta -\phi )+3 \sin (3 \theta +\phi )-3 \sin (\theta +3 \phi )-12 \sin (\theta -3 \phi )+\sin (\theta -5 \phi ))}{4 (\cos (2 \theta )+\cos (2 \phi )-2) (2 \cos (2 \theta )+\cos (2 \phi )-3)} \\
 -\frac{(n-1) u_1 \csc ^5(\phi ) (24 \sin (\theta -\phi )+4 \sin (3 (\theta -\phi ))-12 \sin (3 \theta -\phi )+\sin (5 \theta -\phi )+3 \sin (3 \theta +\phi )-3 \sin (\theta +3 \phi )-12 \sin (\theta -3 \phi )+\sin (\theta -5 \phi ))}{16 \left(2 \sin ^4(\theta ) \csc ^4(\phi )+3 \sin ^2(\theta ) \csc ^2(\phi )+1\right)} \\
 u_1 \\
 -\frac{u_1 \sin (\theta ) \sin (\phi ) \left(\csc ^4(\theta )-\csc ^4(\phi )\right)}{\csc ^2(\theta )+2 \csc ^2(\phi )} \\
 -\frac{u_1 \left(\sin ^2(\theta ) \csc ^2(\phi )+2\right)}{2 \sin ^2(\theta ) \csc ^2(\phi )+1} \\
\end{array}
\right),
\]

\[
r_2 =
\left(
\begin{array}{ccccc}
 \frac{(n-1) u_2 \csc (\theta ) (24 \sin (\theta -\phi )+4 \sin (3 (\theta -\phi ))-12 \sin (3 \theta -\phi )+\sin (5 \theta -\phi )+3 \sin (3 \theta +\phi )-3 \sin (\theta +3 \phi )-12 \sin (\theta -3 \phi )+\sin (\theta -5 \phi ))}{4 (\cos (2 \theta )+\cos (2 \phi )-2) (2 \cos (2 \theta )+\cos (2 \phi )-3)} \\
 \frac{(n-1) u_2 \csc ^5(\phi ) (24 \sin (\theta -\phi )+4 \sin (3 (\theta -\phi ))-12 \sin (3 \theta -\phi )+\sin (5 \theta -\phi )+3 \sin (3 \theta +\phi )-3 \sin (\theta +3 \phi )-12 \sin (\theta -3 \phi )+\sin (\theta -5 \phi ))}{16 \left(2 \sin ^4(\theta ) \csc ^4(\phi )+3 \sin ^2(\theta ) \csc ^2(\phi )+1\right)} \\
 u_2 \\
 -\frac{u_2 \sin (\theta ) \sin (\phi ) \left(\csc ^4(\theta )-\csc ^4(\phi )\right)}{\csc ^2(\theta )+2 \csc ^2(\phi )} \\
 -\frac{u_2 \left(\sin ^2(\theta ) \csc ^2(\phi )+2\right)}{2 \sin ^2(\theta ) \csc ^2(\phi )+1} \\
\end{array}
\right),
\]

\[
r_3 =
\left(
\begin{array}{ccccc}
 0 \\
 0 \\
 u_3 \\
 \frac{u_3 \left(2 \sin ^3(\theta ) \cos (\theta ) \csc ^2(\phi )-\sin (2 \phi )\right)}{\sin (\theta +\phi )+4 \sin ^2(\theta ) \cos (\theta ) \csc (\phi )} \\
 -\frac{u_3 \sin (\theta ) \csc (\phi ) (\sin (\theta ) \csc (\phi ) \sin (\theta +\phi )+2 \sin (2 \phi ))}{\sin (\theta +\phi )+4 \sin ^2(\theta ) \cos (\theta ) \csc (\phi )} \\
\end{array}
\right),
\]
where $u_1, u_2, u_3 \in \mathbbm{R}$ are free parameters.
To fix the orientation of $u_i$, i.e. $u_i \leq 0$ or $u_i \geq 0$, we need to solve the equation $h_k r_i \leq 0$, where $h_k$ is the remaining hyperplane not used in Eq. \eqref{eq:TwoBodyRaysSystems}.
As these rays are conically independent,
we have thus found all rays that characterize the cone of the two-body permutationally invariant Bell operator $\mathbbm{S}^{n,K}_{\varphi,\theta}$, as defined in Eq. \eqref{eq:PolPrimal}.

\section{Two-body cone optimization}\label{app:TwoBodyConeOpt}

The expression of the rays and lines at measurement parameters $(\varphi, \theta)= (\pi/6,5\pi/6)$ are
\begin{equation}
\label{eq:I6RaysAndLinesApp}
\begin{aligned}
    \vec{r}_1&=(-\sqrt{3} (n-1),0,1,-1,0),\\
    \vec{r}_2&=(-\sqrt{3} (n-1),0,-1,1,0),\\
    \vec{r}_3&=(0,0,0,-1,0),\\
    \vec{l}_1&=(-1,1,0,0,0),\\
    \vec{l}_2&=(0,0,1,-2,1).
\end{aligned}
\end{equation}
The numerical values of the rays and lines of Eq. \eqref{eq:I6RaysAndLinesApp} depend on the values of $n$. 
As the absolute value of largest elements of $r_1$ and $r_2$ increase linearly with $n$, we have imposed the following additional constraints on the Bell coefficients for numerical stability during the optimization:
\begin{equation}
\begin{aligned}
\label{eq:RayAndLineConstraints}
    0 \leq c_1 \leq 1/n, \;  0 \leq &c_2 \leq 1/n, \; 0 \leq c_3 \leq 1, \\
    \; -1 \leq c_4 \leq 1, &\; -1 \leq c_5 \leq 1.
\end{aligned}
\end{equation}
Since $\beta_Q/\beta_C$ is invariant under positive rescaling of $\vec{\alpha}$, 
restricting the $c_i$'s to the bounded interval of Eq.~\eqref{eq:RayAndLineConstraints} 
entails no loss of generality.
Any stoquastic Bell operator outside this region can 
be rescaled to lie within it without changing the quantum-classical gap. 
The Bell coefficients $\alpha = (-2,0,1/2,1,-1/2)$ correspond to the Bell operator of Eq. \eqref{eq:Ineq6Op}.
Numerically optimizing $\vec{c}$ for the quantum-classical gap $\beta_Q/\beta_C$, we obtain the coefficients listed in Table \ref{tb:OptGap}.

\section{Three-body lines}\label{app:ThreeBodyLines}

Similarly to the two-body polyhedron, the lines of the three-body polyhedron can be obtained through solving the kernel of $\mathbf{H}$, where each row is of the form
\begin{equation}
\begin{aligned}
\vec{\mathbf{S}}_{k,k+d} = \big( (\mathbf{S}_0)_{k,k+d}, (\mathbf{S}_1)_{k,k+d},(\mathbf{S}_{00})_{k,k+d}, (\mathbf{S}_{01})_{k,k+d},(\mathbf{S}_{11})_{k,k+d},(\mathbf{S}_{000})_{k,k+d},
\dots,
(\mathbf{S}_{111})_{k,k+d} \big).
\nonumber
\end{aligned}
\end{equation}
where $k \in \{0,n-d\}$ and $d \in \{1,2,3\}$.
$\mathbf{H}$ is thus a $(3n-3) \times 9$ matrix. 
We fix $n=10$ and obtain the following 
\begin{align*}
    \vec{l}_1 = (-\csc(\phi) \sin(\theta) , 1 , 0 , 0 , 0, 0,0,0,0), 
\end{align*}
\begin{align*}
        \vec{l}_2 = (0,0,\sin ^2(\theta ) \csc ^2(\phi ),-2 \sin (\theta ) \csc (\phi ),1, 0,0,0,0), 
\end{align*}
\begin{equation}
\begin{aligned}
    \vec{l}_3 = (0,0,0,0,0,&\sin ^3(\theta ) \left(-\csc ^3(\phi )\right),3 \sin ^2(\theta ) \csc ^2(\phi ),-3 \sin (\theta ) \csc (\phi ),1), \nonumber
\end{aligned}
\end{equation}
We have verified that $\vec{l}_1, \vec{l}_2, \vec{l}_3$ indeed are valid lines for all values of $n,\varphi,\theta$, by checking that $\vec{\mathbf{S}}_{k,k+d} \cdot \vec{l_j} = 0 \quad \forall \; k \in \{0,n-d\}, d \in \{1,2,3\}, j \in \{1,2,3\}$.

\section{Classical Bound}

The optimal classical strategy is deterministic. 
Since the number of deterministic strategies is finite, the optimal strategy can be found through exhaustive enumeration.
Due to permutation invariance, only the number of parties choosing a certain strategy is relevant for the value of the classical bound, not which ones. 
Therefore, we can parametrize all possible measurement outcomes by the following four parameters~\cite{tura_nonlocality_2015, guo_detecting_2023}
\begin{align}
    & a = | \{i\in \{0,\cdots,n-1\} \,|\, \langle \mathbf{M}_0^{i} \rangle=1, \langle \mathbf{M}_1^{i} \rangle=1\}| , \nonumber \\  
    & b = | \{i\in \{0,\cdots,n-1\} \,|\, \langle \mathbf{M}_0^{i} \rangle=1, \langle \mathbf{M}_1^{i} \rangle=-1\} |, \nonumber \\ 
    & c = | \{i\in \{0,\cdots,n-1\} \,|\, \langle \mathbf{M}_0^{i} \rangle=-1,\langle \mathbf{M}_1^{i} \rangle=1\} |, \nonumber \\ 
    & d = | \{i\in \{0,\cdots,n-1\} \,|\, \langle \mathbf{M}_0^{i} \rangle=-1,\langle \mathbf{M}_1^{i} \rangle=-1\} | , \nonumber
\end{align}
with the condition $a+b+c+d=n$.
The single-body correlators 
\begin{align}
    \mathcal{S}_{\mu_1} = \sum_{i=0}^{n-1} \langle \mathbf{M}_{\mu_1}^{i} \rangle \;,
    \nonumber
\end{align}
can thus be expressed as
\begin{align}
    \mathcal{S}_0 &= 
    a +b -c -d \;, \nonumber \\
    \mathcal{S}_1 &=a-b+c-d \;. \nonumber
\end{align}
The classical two-body correlator can be rewritten as
\begin{align}\label{eq:decM2nd}
    \mathcal{S}_{\mu_1 \mu_2} &=\sum_{\substack{i_1,i_2=0\\i_1\neq i_2}}^{n-1} \langle \mathbf{M}_{\mu_1}^{i_1}  \mathbf{M}_{\mu_2}^{i_2}\rangle = \sum_{i_1,i_2=0}^{n-1} \langle \mathbf{M}_{\mu_1}^{i_1}  \mathbf{M}_{\mu_2}^{i_2} \rangle - \sum_{i_1=0}^{n-1}  \langle \mathbf{M}_{\mu_1}^{i_1}  \mathbf{M}_{\mu_2}^{i_1} \rangle  \\
    &= \sum_{i_1=0}^{n-1}  \langle \mathbf{M}_{\mu_1}^{i_1}  \rangle \sum_{i_2=0}^{n-1}  \langle \mathbf{M}_{\mu_2}^{i_2}  \rangle - \sum_{i_1=0}^{n-1}  \langle \mathbf{M}_{\mu_1}^{i_1} \rangle \langle \mathbf{M}_{\mu_2}^{i_1} \rangle \;,
\end{align}
where in the last line we have used that the correlators of deterministic strategies factorize.
Expressed in the parameters $a,b,c,d$, we obtain
\begin{align*}
    \mathcal{S}_{00} 
    &= \mathcal{S}^2_0-n, \\
    \mathcal{S}_{11} &= \mathcal{S}^2_1-n, \\
    \mathcal{S}_{01} 
    &= \mathcal{S}_0 \mathcal{S}_1-(a-b-c+d)
\end{align*}
Similarly, the expressions for three-body correlators are
\begin{align*}
    \mathcal{S}_{000} 
    &= \mathcal{S}^3_0+2\mathcal{S}_0-3n\mathcal{S}_0, \\
    \mathcal{S}_{001} 
    &= \mathcal{S}_0\mathcal{S}_0\mathcal{S}_1 +2\mathcal{S}_1 -n\mathcal{S}_1-2(a-b-c+d)\mathcal{S}_0, \\
    \mathcal{S}_{011} 
    &= \mathcal{S}_0\mathcal{S}_1\mathcal{S}_1 +2\mathcal{S}_0-n\mathcal{S}_0-2(a-b-c+d)\mathcal{S}_1,\\
    \mathcal{S}_{111} &= \mathcal{S}^3_1+2\mathcal{S}_1-3n\mathcal{S}_1.
\end{align*}
The classical bound is the calculated as
\begin{align*}
     \beta_C = \min\limits_{\substack{a,b,c,d \\ a+b+c+d=n}}  \alpha_0  \mathcal{S}_0 + \alpha_1  \mathcal{S}_1 + \alpha_2 \mathcal{S}_{00} + \alpha_3  \mathcal{S}_{01} + \alpha_4 \mathcal{S}_{11} +  \alpha_5 \mathcal{S}_{000} +  \alpha_6 \mathcal{S}_{001} +  \alpha_7 \mathcal{S}_{011}+  \alpha_8 \mathcal{S}_{111}.
 \nonumber
\end{align*}

\section{Quantum-classical gap optimization}

In this section, we describe our approach to the optimization of the quantum-classical gap $f_{n,\varphi,\theta}(\vec{c})$ at fixed measurement parameters $(\varphi,\theta)$ and a fixed number of parties $n$. 
Here $\vec{c}$ denotes the vector of coefficients such that $\vec{\alpha} = \sum_i (\vec{c})_i \vec{g}_i$, where $(\vec{c})_i$ is the $i$-th element of $\vec{c}$ and $\vec{g}_i$ is either a ray or a line. 
For concreteness, let us fix the measurement parameters $(\varphi, \theta)=(\pi/4,-\pi/4)$ and the number of parties $n=10$. 
We randomly initialize the coefficient vector $\vec{c}_0 = (c_1,\ldots,c_m)$, where $m$ is the total number of rays and lines characterizing the stoquasticity cone $\mathcal{C}$.
We denote the $i$-th element of the vector $\vec{c}_0$ as $(\vec{c}_0)_i$.
Then for all $i\in \{1,\ldots,m\}$, we sweep over a fixed interval. 
We then select the value $(\vec{c}_0)_i$ that increases the objective function $f_{n,\varphi,\theta}(\vec{c})$ the most, denoted as $(\vec{c}_0)_i^*$. 
We replace $(\vec{c}_0)_i \rightarrow (\vec{c}_0)_i^*$. 
After this replacement, we denote our new vector as $\vec{c}_1$.
We repeat this until the objective function $f_{\varphi,\theta}(\vec{c})$ has converged after $T$ timesteps. 
We note that the sequence $\left(f_{n,\varphi,\theta}(\vec{c}_0),\ldots,f_{n,\varphi,\theta}(\vec{c}_t)\ldots,f_{n,\varphi,\theta}(\vec{c}_T)\right)$, where $\vec{c}_t$ denotes the coefficients after $t$ time steps, is a monotonically increasing sequence. 
The sequence thus necessarily converges to a local optimum.

For our optimization of the quantum-classical gap, we optimized within the stoquasticity cone of $(n,\varphi,\theta)=(10,\pi/4,-\pi/4)$.
After the optimization sequence, which is represented in Fig. \ref{fig:ManOptSeq}, the following Bell coefficients were found:
\begin{figure}[h!]
\centering
\begin{minipage}[b]{0.30\textwidth}
    \centering
    \includegraphics[width=\textwidth]{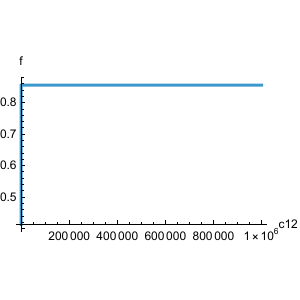}
    (a) $t=1$
\end{minipage}
\hfill
\begin{minipage}[b]{0.30\textwidth}
    \centering
    \includegraphics[width=\textwidth]{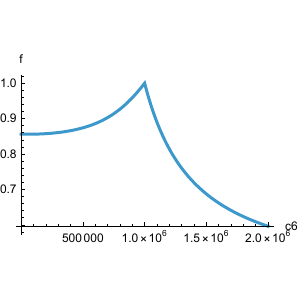}
    (b) $t=2$
\end{minipage}
\hfill
\begin{minipage}[b]{0.30\textwidth}
    \centering
    \includegraphics[width=\textwidth]{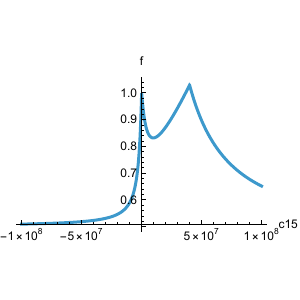}
    (c) $t=3$
\end{minipage}
\caption{Optimization sequence leading to the optimal Bell operator for $(\varphi,\theta)=(\pi/4,-\pi/4)$. For the first two optimization time steps, the domain was set to $[0,10^6]$ for the rays and $[-10^6,10^6]$ for the lines.
For the last optimization step, the domain was increased to $[0,10^8]$ for the rays and $[-10^8,10^8]$ for the lines. The coefficients $c_1,\ldots,c_{13}$ correspond to rays and $c_{14}, c_{15}, c_{16}$ correspond to lines.}\label{fig:ManOptSeq}
\end{figure}
\begin{equation}
\begin{aligned}
    \vec{\alpha} = (&-1.6\times 10^8,0.914772,3.9994\times 10^7,
    7.9988\times 10^7,3.9994\times 10^7,\\
    &-0.179356,-0.343468,-0.22854,-0.045208).
\end{aligned}
\end{equation}
As the quantum-classical gap is invariant under rescaling of $\alpha$, we obtain the following Bell coefficients 
\begin{equation}
    \vec{\alpha} = (-2.0,0,0.5,1,0.5,0,0,0,0),
\end{equation}
after dividing by $8*10^7$ and rounding to one decimal place, which show remarkable similarities compared to Eq. \eqref{eq:Ineq6Op}.
For $n=10,20,30,50$, the corresponding quantum-classical gaps are $f_{(10,\pi/4,-\pi/4)}=1.02904, f_{(20,\pi/4,-\pi/4)}=1.06576, f_{(30,\pi/4,-\pi/4)}=1.08556, f_{(50,\pi/4,-\pi/4)}=1.10791$. 
We would like to make a note that the initial starting point strongly influences the optimization sequence.  
For the parameters $(n,\varphi,\theta)=(8,\pi/4,-\pi/4)$ and a differently chosen initial starting point $\vec{c}_0'$, the resulting sequence does not converge to a quantum-classical gap greater than $1$, cf. Fig. \ref{fig:SubOptSeq}.
\begin{figure}[H]
\centering
\begin{minipage}[b]{0.30\textwidth}
    \centering
    \includegraphics[width=\textwidth]{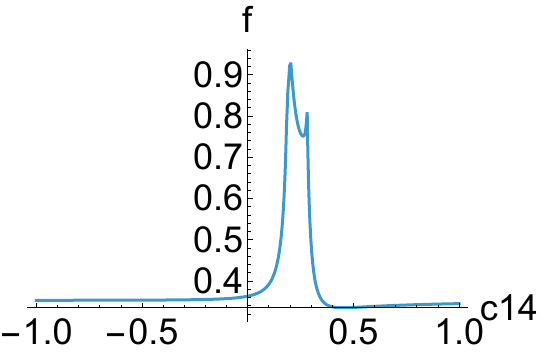}
    (a) $t=1$
\end{minipage}
\hfill
\begin{minipage}[b]{0.30\textwidth}
    \centering
    \includegraphics[width=\textwidth]{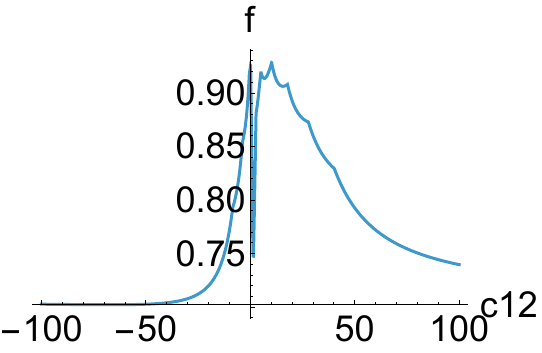}
    (b) $t=2$
\end{minipage}
\hfill
\begin{minipage}[b]{0.30\textwidth}
    \centering
    \includegraphics[width=\textwidth]{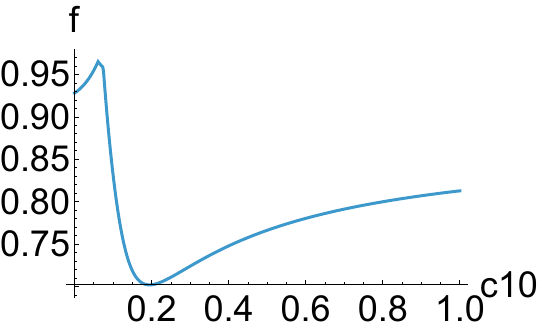}
    (c) $t=3$
\end{minipage}

\medskip

\begin{minipage}[b]{0.30\textwidth}
    \centering
    \includegraphics[width=\textwidth]{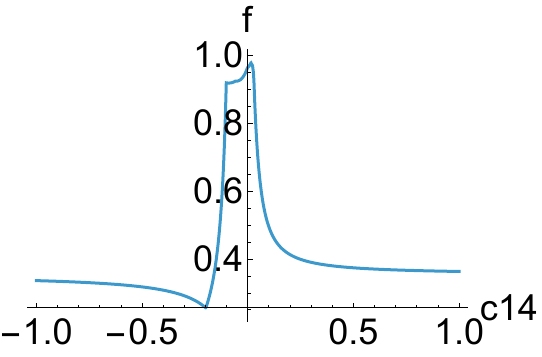}
    (d) $t=4$
\end{minipage}
\hfill
\begin{minipage}[b]{0.30\textwidth}
    \centering
    \includegraphics[width=\textwidth]{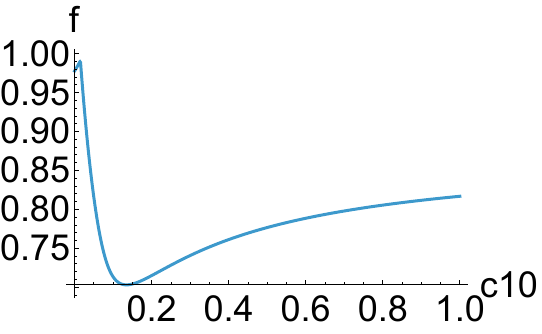}
    (e) $t=5$
\end{minipage}
\hfill
\begin{minipage}[b]{0.30\textwidth}
    \centering
    \includegraphics[width=\textwidth]{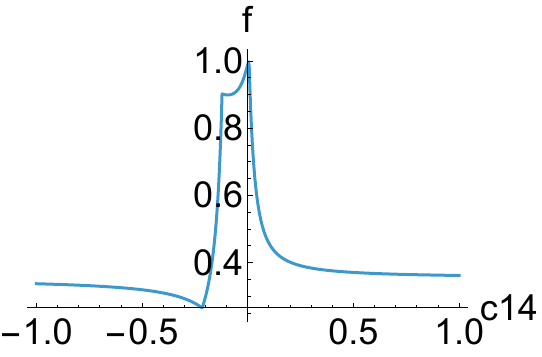}
    (f) $t=6$
\end{minipage}

\caption{Example of optimization sequence converging to $f_{n,\varphi,\theta}(\vec{c}) \leq 1$.}
\label{fig:SubOptSeq}
\end{figure}

\section{Code availability}\label{app:Nums}

The data and code supporting the findings of this study are available from the corresponding author upon reasonable request.

\end{document}